\documentclass[journal,twoside,web]{ieeecolor}
\usepackage{generic}
\usepackage{cite}
\usepackage{amsmath,amssymb,amsfonts}
\usepackage{algorithmic}
\usepackage{graphicx}
\usepackage{textcomp}
\usepackage{algorithm}
\newtheorem{theorem}{Theorem} 
\newtheorem{proposition}{Proposition} 
\newtheorem{corollary}{Corollary} 
 
\newtheorem{lemma}{Lemma}
\newtheorem{remark}{Remark}
\newtheorem{definition}{Definition} 
\newtheorem{assumption}{Assumption}
\newtheorem{example}{Example}
\def\BibTeX{{\rm B\kern-.05em{\sc i\kern-.025em b}\kern-.08em
		T\kern-.1667em\lower.7ex\hbox{E}\kern-.125emX}}
\begin{document}
	\title{Adaptive Identification with Guaranteed Performance Under Saturated-Observation and Non-Persistent Excitation}
	\author{Lantian Zhang and Lei Guo, \IEEEmembership{Fellow, IEEE}
		\thanks{This work was supported by the National Natural Science Foundation of China under Grant No. 12288201.}
		\thanks{Lantian Zhang and Lei Guo are with the Key Laboratory of Systems and Control, Academy of Mathematics and Systems Science, Chinese Academy of Sciences, Beijing 100190, China, and also with the School of Mathematical Science, University of Chinese Academy of Sciences, Beijing 100049, China. (e-mails: zhanglantian@amss.ac.cn, Lguo@amss.ac.cn). }}
	
	\maketitle
	
	\begin{abstract}
		This paper investigates adaptive identification and prediction problems for stochastic dynamical systems with saturated output observations, which arise from various fields in engineering and social systems, but up to now still lack comprehensive theoretical studies including guarantees for the estimation performance needed in practical applications. With this impetus, the paper has made the following main contributions: (i) To introduce an adaptive two-step quasi-Newton algorithm to improve the performance of the identification, which is applicable to a typical class of nonlinear stochastic systems with outputs observed under possibly varying saturation. (ii) To establish the global convergence of both the parameter estimators and adaptive predictors and to prove the asymptotic normality,  under the weakest possible non-persistent excitation condition, which can be applied to stochastic feedback systems with general non-stationary and correlated system signals or data. (iii) To establish useful probabilistic estimation error bounds for any given finite length of data, using either martingale inequalities or Monte Carlo experiments. A numerical example is also provided to illustrate the performance of the proposed identification algorithm. 	
	\end{abstract}
	
	
	\section{Introduction}
	\label{sec:introduction}	
	Identifying the input-output relationship and predicting the future behavior of dynamical systems based on observation data are fundamental problems in various fields including control systems, signal processes, machine learning, etc. This paper considers identification and prediction problems for stochastic dynamical systems with saturated output observation data. Here, by saturated output observations, we mean that the observations for the output are produced through the following mechanism: at each time, the noise-corrupted output can be observed precisely only when its value lies in a certain range, however, when the output value exceeds this range, its observation becomes saturated, leading to imprecise information. The relationship between the system output and its observation is illustrated in Fig.\ref{fig1}, where $v_{k+1}$ and $y_{k+1}$ represent the system output and its observation respectively, the interval $[l_{k}, u_{k}]$ is the precise observation range, when the system output exceeds this range,  the only possible observation is a constant,  either $L_{k}$ or $U_{k}$. Note that if we take  $L_{k}=l_{k}=0, \;u_{k}=U_{k}=\infty$, then the saturation function will become the  ReLu function widely used in machine learning; and if we take $L_{k}=l_{k}=u_{k}=0,\; U_{k}=1$, the saturation function will turn to be a binary-valued function widely used in classification problems(\cite{mc1943}, \cite{gs1990}).
	\begin{figure}[htbp]
		\centerline{\includegraphics[width=7cm]{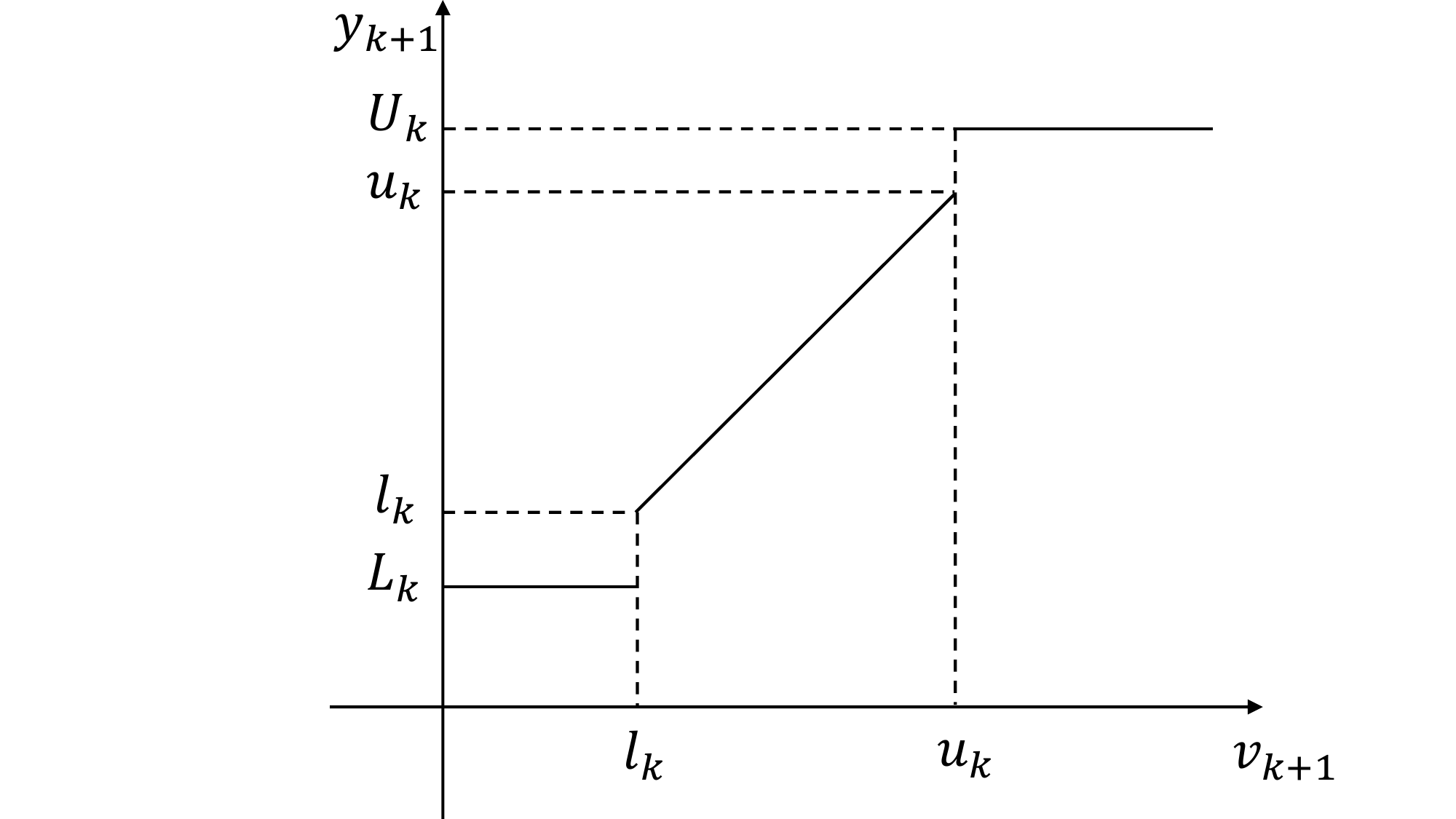}}
		\caption{Saturated output observations.}
		\label{fig1}
	\end{figure}
	
	Saturated output observations in stochastic dynamical systems exist widely in various fields including engineering (\cite{sj2004}\cite{hf2009}), economics (\cite{tj1958}-\cite{cm2013}), and social systems\cite{judical}.  We only mention  several examples in three different
	application areas.
	The first example is from control engineering (\cite{sj2004}), where $y_{k}$ represents the sensor observation of the system output, which can be considered as a saturated output observation since it becomes saturated if the output is too large to exceed the observation range of the available sensors. The second example is from economics (\cite{tj1958}), where $v_{k}$ is interpreted as an index of the consumer's intensity of desire to purchase a durable, $y_{k}$ is the true purchase which can be regarded as a saturated observation, since the intensity $v_{k}$ can be observed only if it exceeds a certain threshold where the true purchase takes place; The third example is from sentencing (\cite{judical}), where $y_{k}$ is the pronounced penalty and can also be regarded as a saturated observation since it is constrained within the statutory range of penalty according to the related basic criminal facts.

Since the emergence of saturation changes the structure of the original systems and may degrade system performance, providing a theoretical analysis for the performance guarantee of the identification algorithm is one of the most important issues addressed in this paper. Compared with the unsaturated case of observations, the key challenge of the current saturated case is the inherent nonlinearity in the observations of the underlying stochastic dynamical systems. In the past several decades, various identification methods with saturated observations have been studied intensively  with  both asymptotic and non-asymptotic results, on which we give a brief review separately in the following:
	
	First, most of the existing theoretical results are asymptotic in nature, where the number of observations needs to increase unboundedly or at least to be sufficiently large. For example,  the least absolute deviation methods were considered in \cite{PJ1984}, and the strong consistency and asymptotic normality of the estimators were proven for independent signals satisfying the usual persistent excitation (PE) condition where the condition number of the information matrix is bounded. Besides, the maximum likelihood (ML) method was considered in \cite{ly1992}, where consistency and asymptotic efficiency were established for independent or non-random signals satisfying a stronger PE condition. Moreover, a two-stage procedure based on ML was proposed in \cite{hj1976} to deal with two coupled models with saturated observations. Furthermore, the empirical measure approach was employed in \cite{ZG2003}-\cite{YG2007}, where strong consistency and asymptotic efficiency were established under periodic signals with binary-valued observations.  Such observations were also considered in \cite{GZ2013}, where a strongly consistent recursive projection algorithm was given under a condition stronger than the usual PE condition.  
	
	Second, there are also a number of non-asymptotic estimation results in the literature. Despite the importance of the asymptotic estimation results as mentioned above,  non-asymptotic results appear to be more practical, because one usually only has a finite number of data available for identification in practice. 
	However, obtaining non-asymptotic identification results, which are usually given under high probability, is quite challenging especially when the structure comes to nonlinear. 
Most of the existing results are established under assumptions that the system data are independently and identically distributed (i.i.d), e.g., the analysis of the stochastic gradient descent methods in \cite{kv2019},\cite{dc2021}. For the dependent data case, an online Newton method was proposed in \cite{po2021},  where a probabilistic error bound was given for linear stochastic dynamical systems where the usual PE  condition is satisfied.
	
	In summary, almost all of the existing identification results for stochastic dynamical systems with saturated observations need at least the usual PE condition on the system data, and actually, most need i.i.d assumptions. Though these idealized conditions are convenient for theoretical investigation, they are hardly satisfied or verified for general stochastic dynamical systems with feedback signals (see, e.g. \cite{1412020}).  This inevitably brings challenges for establishing an identification theory on either asymptotic or non-asymptotic results with saturated observations under more general (non-PE) signal conditions.

	Fortunately, there is a great deal of research on adaptive identification for linear or nonlinear stochastic dynamical systems with uncensored or unsaturated observations in the area of adaptive control, where the system data used can include those generated from stochastic feedback control systems. By adaptive identification, we mean that the identification algorithm is constructed recursively,  where the parameter estimates are updated online based on both the current estimate and the new observation, and thus the iteration instance depends on the time of the data observed. In comparison with offline algorithms such as those widely used in statistics and machine learning where the iteration instance is the number of search steps in numerical optimization (\cite{yy2015},\cite{sm2018}), the adaptive algorithm has at least two advantages: one is that the algorithm can be updated conveniently when new data come in without restoring the old data, and another is that general non-stationary and correlated data can be handled conveniently due to the structure of the adaptive algorithm. In fact, extensive investigations have been conducted in adaptive identification in the area of control systems for the design of adaptive control laws, where the system data are generated from feedback systems that are far from stationary and hard to be analyzed  \cite{g1995}.  Many adaptive identification methods have been introduced in the existing literature where the convergence has also been analyzed under certain non-PE conditions (see e.g. \cite{cg1991}-\cite{lw1982}). Among these methods, we mention that Shadab et al. \cite{ss2023} considered the first-order gradient estimator for linear regression models with some finite-time parameter estimation techniques, where the PE condition is replaced with a condition enforced to the determinant of an extended regressor matrix. Ljung \cite{Lj1977} established a convergence theory via the celebrated ordinary differential equation (ODE) method which can be applied to a wide class of adaptive algorithms, where the conditions for regressors are replaced by some stability conditions for the corresponding ODE. Lai and Wei \cite{lw1982} considered the classical least squares algorithm for linear stochastic regression models, and established successfully the strong consistency under the weakest possible non-PE condition. Of course, these results are established for the traditional non-saturated observation case.
	
The first paper that establishes the strong consistency of estimators for stochastic regression models under general non-PE conditions for a special class of saturated observations (binary-valued observations) appears to be \cite{ZZ2022}, where a single-step adaptive quasi-Newton-type algorithm was proposed and analyzed. The non-PE condition used in \cite{ZZ2022}  is similar to the weakest possible signal condition for a stochastic linear regression model with uncensored observations (see \cite{lw1982}), which can be applied to non-stationary stochastic dynamical systems with feedback control. However,  there are still some unresolved fundamental problems, for instance, a) How should a globally convergent estimation algorithm be designed for stochastic systems under general saturated observations and non-PE conditions?  b) What is the asymptotic distribution of the estimation error under non-PE conditions? c) How to get a  useful and computable probabilistic estimation error bound under non-PE conditions when the length of data is finite?

	The main purpose of this paper is to solve these problems by introducing an adaptive two-step quasi-Newton-type identification algorithm, refining the stochastic Lyapunov function approach, and applying some martingale inequalities and convergence theorems.  Besides, the Monte Carlo method is also found quite useful in computing the estimation error bound. The key feature of the current adaptive two-step quasi-Newton (TSQN) identification algorithm compared to the adaptive single-step quasi-Newton method is that the TSQN algorithm has improved performance under non-PE conditions. The main reasons behind this fact are as follows: (i) The scalar adaptation gain of the single-step quasi-Newton method is constructed by using only the fixed given {\it a priori} information about the parameter set, whereas the scalar adaptation gain of the current TSQN algorithm is designed by using the online information to have it improved adaptively. (ii) A regularization factor is also introduced in the TSQN algorithm, which can be taken as a ``noise variance" estimate to improve the adaptation algorithm further. To be specific, the main contributions of this paper can be summarized as follows:
	
	\begin{itemize}
		\item A new two-step quasi-Newton-type adaptive identification algorithm is proposed for stochastic dynamical systems with saturated observations.  The first step is to produce consistent parameter estimates based on the available ``worst case'' information, which is then used to construct the adaptation gains in the second step for improving the performance of the adaptive identification. 
		
		\item Asymptotic results on the proposed new identification algorithm, including strong consistency and asymptotic normality,  are established for stochastic dynamical systems with saturated observations under quite general non-PE conditions. The optimality of adaptive prediction or regrets is also established without resorting to any excitation conditions. This paper appears to be the first for adapted identification of stochastic systems with guaranteed performance under both saturated observations and general non-PE conditions.
				
		\item Non-asymptotic error bounds for both parameter estimation and output prediction are also provided, when only a finite length of data is available,  for stochastic dynamical systems with saturated observations and no PE conditions. Such bounds can be applied to sentencing computation problems based on practical judicial data \cite{judical}.
		
	\end{itemize}
	
	The remainder of this paper is organized as follows. Section \ref{ss2} gives the problem formulation; The main results are stated in Section \ref{ss3}; Section \ref{ss4} presents the proofs of the main results. A numerical example is provided in Section \ref{ss5}. Finally, we conclude the paper with some remarks in Section \ref{ss6}.
	
	\section{Problem Formulation}\label{ss2}
	
	Let us consider the following piecewise linear regression model: 
	
	\begin{equation}\label{eq1}
		y_{k+1}=S_{k}(\phi_{k}^{\top}\theta+e_{k+1}),\;\;k=0,1,\cdots,	
	\end{equation}
	where  $\theta\in \mathbb{R}^{m} (m\geq1)$ is an unknown parameter vector to be estimated;  $y_{k+1}\in \mathbb{R}$, $\phi_{k}\in \mathbb{R}^{m}$, $e_{k+1}\in \mathbb{R}$ represent the system output observation, stochastic regressor, and random noise at time $k$, respectively. Besides, $S_{k}(\cdot):\mathbb{R} \rightarrow \mathbb{R}$ is a non-decreasing time-varying saturation function defined as follows: 
	\begin{equation}\label{eq2}
		S_{k}(x)=\left\{
		\begin{array}{rcl}
			&L_{k}             & {x     <     l_{k}}\\
			&x         & {l_{k} \leq x \leq u_{k}}\\
			&U_{k}      & {x > u_{k}}
		\end{array} \right.,\;\;\;k=0,1,\cdots,
	\end{equation}
	where $[l_{k}, u_{k}]$ is the given precise observable range,  $L_{k}$ and  $U_{k}$ are the only observations when the output value exceeds this range.

	\subsection{Notations and Assumptions}	
	{\bf Notations.}  By $\|\cdot\|$, we denote the Euclidean norm of vectors or matrices. The spectrum of a matrix $M$ is denoted by $\left\{\lambda_{i}\left\{M\right\}\right\}$, where the maximum and minimum eigenvalues  are denoted by $\lambda_{max}\left\{M\right\}$ and $\lambda_{min}\left\{M\right\}$  respectively. Besides, let $tr(M)$ denote the trace of the matrix $M$, and by $|M|$  we mean the determinant of the matrix $M$. Moreover, $\left\{\mathcal{F}_{k},k\geq 0\right\}$ is the sequence of $\sigma -$algebra together  with that of conditional mathematical expectation operator $\mathbb{E}[\cdot \mid \mathcal{F}_{k}]$, in the sequel we may employ the abbreviation $\mathbb{E}_{k}\left[\cdot\right]$ to $\mathbb{E}\left[\cdot \mid \mathcal{F}_{k}\right]$. Furthermore, a random variable $X$ belongs to $\mathcal{L}_{2}$ if $\mathbb{E}\|X\|^{2}<\infty,$ and a random sequence $\{X_{k}, k\geq 0\}$ is called $\mathcal{L}_{2}$ sequence if $X_{k}$ belongs to $\mathcal{L}_{2}$ for all $k\geq 0$.
	
	We  need the following basic assumptions:
	\begin{assumption}\label{assum1}
		The stochastic  regressor  $\{\phi_{k},\mathcal{F}_{k}\}$ is  a bounded and adapted sequence, where $\left\{\mathcal{F}_{k},k\geq 0\right\}$ is a non-decreasing sequence of $\sigma -$algebras. Besides, the true parameter $\theta$ is an interior point of a known convex compact set $D \subseteq 	\mathbb{R}^{m} $. 
	\end{assumption}
	
	By Assumption $\ref{assum1}$, we can find an almost surely bounded sequence $\{M_{k}, k\geq 0\}$ such that 
	\begin{equation}\label{Mk}
	\sup\limits_{x\in D}|\phi_{k}^{\top}x| \leq M_{k},\;\;a.s.
	\end{equation}

	\begin{assumption}\label{assum2}
		The thresholds $\left\{l_{k},\mathcal{F}_{k}\right\}$, $\left\{u_{k},\mathcal{F}_{k}\right\}$, $\left\{L_{k},\mathcal{F}_{k}\right\}$ and $\left\{U_{k},\mathcal{F}_{k}\right\}$ are known adapted stochastic sequences, satisfying for any $k\geq 0$,
		\begin{equation}\label{5}
			l_{k}-c \leq L_{k}\leq l_{k}\leq u_{k}\leq U_{k}\leq u_{k}+ c,\;\;\;a.s.,	
			\end{equation}	
			where $c$ is a $\mathcal{L}_{2}$ non-negative random variable, and
			\begin{equation}\label{6}
			\sup_{k\geq 0} l_{k}<\infty, \;\;\inf_{k\geq 0} u_{k}>-\infty,\;\;\;a.s.
			\end{equation}			
	\end{assumption}
	\begin{remark}
		We note that the inequalities $L_{k}\leq l_{k} \leq u_{k} \leq U_{k}$ are determined by the non-decreasing nature of the saturation function used to characterize the saturated output observations as illustrated in Fig. 1, and that Assumption $\ref{assum2}$ will be automatically satisfied if $\{L_{k}\}$ and $\{U_{k}\}$ are bounded stochastic sequences. The conditions $(\ref{5})$ and $(\ref{6})$ are general assumptions that are used to guarantee the boundedness of the variances of the output prediction errors in the paper.
	\end{remark}

	\begin{assumption}\label{assum4}
		The noise $\{e_{k},\mathcal{F}_{k}\}$ is an $\mathcal{L}_{2}$ martingale difference sequence and there exists a constant $\eta>0$, such that
				\begin{equation}\label{e}	
		\inf\limits_{k\geq0}\mathbb{E}_{k}\left[|e_{k+1}|^{2}\right]>0,\;\;	\sup_{k\geq 0}\mathbb{E}_{k}\left[|e_{k+1}|^{2+\eta}\right]<\infty,\;a.s.
		\end{equation}		
		Besides, the conditional expectation function $G_{k}(x)$, defined by $G_{k}(x)=\mathbb{E}_{k}\left[S_{k}(x+e_{k+1})\right]$, is known and differentiable with derivative denoted by $G'_{k}(\cdot)$. Further, there exist a random variable $M>\sup\limits_{k\geq 0} M_{k}$ such that 
		\begin{equation}\label{iinf}
			0<\inf _{|x| \leq M, k\geq 0} G'_{k}(x)\leq \sup _{|x| \leq M, k\geq 0} G'_{k}(x)<\infty, \;\;a.s.
		\end{equation}
			\begin{equation}\label{iinft}
 |G'_{k}(x)-G'_{k}(y)|\leq \rho|x-y|,\;a.s., \;\;\;\forall |x|,|y|\leq M,
	\end{equation}
where $\rho$ is a non-negative variable, $M_{k}$ is defined in $(\ref{Mk})$.
		\end{assumption}

	\begin{remark}	
		It is worth to mention that under condition $(\ref{e})$, 	
		the function $G_{k}(\cdot)$ in Assumption $\ref{assum4}$ is well-defined for any $k\geq 0$, and can be calculated given the conditional probability distribution of the noise $e_{k+1}$.
		In Appendix \ref{i}, we have provided three typical examples to illustrate how to concretely calculate the function $G_{k}(\cdot)$, which includes the classical linear stochastic regression models, models with binary-valued sensors, and censored regression models. 
Moreover,	 Assumption $\ref{assum4}$  can be easily verified if the noise $\{e_{k+1}, \; k\geq 0\}$ is i.i.d Gaussian and if $\inf\limits_{k\geq 0}(U_{k}-L_{k})>0,\;a.s$.
		Besides, when  $l_{k}=-\infty$ and $u_{k}=\infty$, the system $(\ref{eq1})$-$(\ref{eq2})$ will degenerate to linear stochastic regression models,  and Assumption $\ref{assum4}$ will degenerate to the standard noise assumption for the strong consistency of the classical least squares (\cite{lw1982}) since $G_{k}(x)\equiv x$. 		
	\end{remark}	

		For simplicity of notation, denote
	\begin{equation}\label{g}
		\inf _{|x| \leq M_{k}} G'_{k}(x)=\underline{g}_{k}, \;\;
		\sup _{|x| \leq M_{k}} G'_{k}(x)=\overline{g}_{k}.
	\end{equation}
	Under Assumption $\ref{assum4}$,  $\{\overline{g}_{k}, k\geq 0\}$ and $\{\underline{g}_{k}, k\geq 0\}$ have upper bound and positive lower bound respectively, i.e. 
	\begin{equation}\label{zh}
		\inf\limits_{k\geq 0}\underline{g}_{k}>0,\;\;\;\;\;\;
	\sup\limits_{k\geq 0} \overline{g}_{k}<\infty,\;\;\;\;a.s.
	\end{equation}

	\subsection{Algorithm}
	
	Because of its ``optimality" and fast convergence rate, the classical LS algorithm is one of the most basic and widely used ones in the adaptive estimation and adaptive control of linear stochastic systems. Inspired by the analysis of the LS recursive algorithm, we have introduced an adaptive quasi-Newton-type algorithm to estimate the parameters in linear stochastic regression models with binary-valued observations in \cite{ZZ2022}. However, we find that a direct extension of the quasi-Newton algorithm introduced in \cite{ZZ2022} from binary-valued observation to saturated observation does not give satisfactory performance, which motivates us to introduce a two-step quasi-Newton-type identification algorithm as described shortly.
	
	At first, we introduce a suitable projection operator, to ensure the boundedness of the estimates while keeping other nice properties. For  the linear space $\mathbb{R}^{m}$, we define a norm $\|\cdot \|_{Q}$ associated with a positive definite matrix $Q$  as $\| x\|_{Q}^{2}=x^{\top}Qx$.
 	 A projection operator based on $\| \cdot\|_{Q}$ is defined as follows:
	\begin{definition}\label{def2}
		For the  convex compact set $D$ defined  in Assumption  $\ref{assum1}$,  the projection operator $\Pi_{Q}(x)(\cdot)$ is defined as
		\begin{equation}\label{8}
			\Pi_{Q}(x)=\mathop{\arg\min}_{y \in D}\|x-y\|_{Q}, \quad \forall x \in 	\mathbb{R}^{m}.
		\end{equation}
	\end{definition}

	We then introduce our new adaptive two-step quasi-Newton (TSQN) identification algorithm,	 where the gain matrix is constructed by using the gradient information of the quadratic loss function.
	\begin{algorithm}[htb]  %调整间距
		\caption{ Adaptive Two-Step Quasi-Newton (TSQN) Algorithm } 
		\label{alg1} 	
		\textbf{Step 1.} Recursively calculate the preliminary estimate $\bar{\theta}_{k+1}$ for $k\geq 0$:
		\begin{equation}\label{be1}
			\begin{aligned}
				\bar{\theta}_{k+1}&=\Pi_{\bar{P}_{k+1}^{-1}}\{\bar{\theta}_{k}+\bar{a}_{k}\bar{\beta}_{k} \bar{P}_{k}\phi_{k}[y_{k+1}-G_{k}(\phi_{k}^{\top}\bar{\theta}_{k})]\},\\
				\bar{P}_{k+1}&=\bar{P}_{k}-\bar{a}_{k}\bar{\beta}_{k}^{2}\bar{P}_{k}\phi_{k}\phi_{k}^{\top}\bar{P}_{k},\\
				\bar{\beta}_{k}&=\min\left(\underline{g}_{k}, \frac{1}{1+2\overline{g}_{k}\phi_{k}^{\top}\bar{P}_{k}\phi_{k}}\right),\\						\bar{a}_{k}&=\frac{1}{1+\bar{\beta}_{k}^{2}\phi_{k}^{\top}\bar{P}_{k}\phi_{k}},
			\end{aligned}
		\end{equation}
		where  $\underline{g}_{k}$ and $\overline{g}_{k}$ are defined as in $(\ref{g})$, $\Pi_{\bar{P}_{k+1}^{-1}}$ is the projection operator defined as in Definition $\ref{def2}$, $G_{k}(\cdot)$ is defined in Assumption $\ref{assum4}$, 
		the initial values $\bar{\theta}_{0}$ and $\bar{P}_{0}$ can be chosen arbitrarily in $D$ and with $\bar{P}_{0}>0$,  respectively.
		
		\textbf{Step 2.} Recursively define  the accelerated estimate $\hat{\theta}_{k+1}$ based on $\bar{\theta}_{k+1}$ for $k\geq 0$:
		\begin{equation}\label{be2}
			\begin{aligned}
				\hat{\theta}_{k+1}=&\Pi_{P_{k+1}^{-1}}\{\hat{\theta}_{k}+a_{k}\beta_{k} P_{k}\phi_{k}[y_{k+1}-G_{k}(\phi_{k}^{\top}\hat{\theta}_{k})]\},\\
				P_{k+1}=&P_{k}-a_{k}\beta_{k}^{2}P_{k}\phi_{k}\phi_{k}^{\top}P_{k},\\
				\beta_{k}=&\frac{G_{k}(\phi_{k}^{\top}\overline{\theta}_{k})-G_{k}(\phi_{k}^{\top}\hat{\theta}_{k})}{\phi_{k}^{\top}\overline{\theta}_{k}-\phi_{k}^{\top}\hat{\theta}_{k}}I_{\{\phi_{k}^{\top}\hat{\theta}_{k}-\phi_{k}^{\top}\overline{\theta}_{k}\not=0\}}\\
				&+G'_{k}(\phi_{k}^{\top}\hat{\theta}_{k})I_{\{\phi_{k}^{\top}\hat{\theta}_{k}-\phi_{k}^{\top}\overline{\theta}_{k}=0\}},\\
				a_{k}=&\frac{1}{\mu_{k}+\beta_{k}^{2}\phi_{k}^{\top}P_{k}\phi_{k}},
			\end{aligned}
		\end{equation}
		where $\{\mu_{k}\}$ can be any positive random process adapted to  $\{\mathcal{F}_{k}\}$ with $0<\inf\limits_{k\geq 0} \mu_{k}\leq \sup\limits_{k\geq 0}\mu_{k}<\infty$, the initial values $\hat{\theta}_{0}$ and $P_{0}$ can be chosen arbitrarily in $D$ and with $P_{0}>0$,
		respectively.
	\end{algorithm}	
	
	\begin{remark}
		As described above, our identification algorithm is defined by two successive steps, between which the main difference is the construction of the adaptation gains. In the first step, the scalar adaptation gain $\bar{\beta}_{k}$ is constructed by using the bounds $\underline{g}_{k}$ and $\overline{g}_{k}$ defined in $(\ref{g})$, in a similar way as that constructed in the identification algorithm of \cite{ZZ2022}. Though the strong consistency of the preliminary estimate $\bar{\theta}_{k}$ in the first step may be established following a similar argument as in  \cite{ZZ2022}, its convergence speed appears to be not good enough, and its asymptotic normality also appears to be hard to establish, because the scalar adaptation gain $\bar{\beta}_{k}$ is simply constructed by using the  ``worst case" information $\underline{g}_{k}$ and $\overline{g}_{k}$.  To overcome these shortcomings, the second step estimation is introduced with the following two features:  (i) To improve the performance of the estimation algorithm, the scalar adaptation gain $\beta_{k}$ is defined in an adaptive way by using the preliminary estimates $\bar{\theta}_{k}$ generated in the first step, and (ii) To ensure the asymptotic normality of the estimation errors under non-PE condition, the regularization factor $\mu_{k}$ is taken as a  ``noise variance" estimate constructed by using the online estimates (see Theorem $\ref{thm3}$).  Simulations in Section 5 also demonstrate that the convergence speed of the parameter estimates given in the second step outperforms that of the first step. 
	\end{remark}
	
	\section{Main results}\label{ss3}
	In this section, we give some asymptotic results of the TSQN identification algorithm.  To be specific, we will establish asymptotic upper bounds for both the parameter estimation errors and the adaptive prediction errors in Subsection \ref{sub1},  study the asymptotic normality of the TSQN algorithm in Subsection \ref{sub2}, and give high probabilistic error bounds for any given finite number of data in Subsection \ref{sub3}. 
	
	\subsection{Asymptotic error bounds}\label{sub1}		
	\begin{theorem} \label{thm2}
		Under Assumptions $\ref{assum1}$-$\ref{assum4}$,  the estimate $\hat{\theta}_{k}$ given by the TSQN Algorithm  has the following upper bound almost surely as $k\rightarrow \infty$:
		
		\begin{equation}\label{theta2}
			\|\tilde{\theta}_{k+1}\|^{2}=O\left(\frac{\log \lambda_{max}(k)}{\lambda_{min}(k)}\right),\;\;\;a.s.,
		\end{equation}			
		where $\tilde{\theta}_{k+1}=\theta-\hat{\theta}_{k+1}$, $\lambda_{\min}(k)$ and $\lambda_{\max}(k)$ are the minimum and maximum eigenvalues of  the matrix $\sum\limits_{i=0}^{k}  \phi_{i} \phi_{i}^{\top}+P_{0}^{-1}$ respectively.  
	\end{theorem}
	\begin{remark}\label{re77}
	From Theorem \ref{thm2}, it is easy to see that the algorithm will  converge to the true parameter almost surely if 
	\begin{equation}\label{re}
		\frac{\log \lambda_{max}(k)}{\lambda_{min}(k)}=o(1),\;\;a.s.
	\end{equation}
	This condition does not need the independence and stationarity conditions on the system regressors and is hence applicable to stochastic  feedback control systems, and is known to be the weakest possible convergence condition of the classical least squares in the linear case (see \cite{lw1982},\cite{cg1991}), which is much weaker than the traditional
	PE condition,  i.e. $\lambda_{\max}(k)=O(\lambda_{\min}(k))$. Furthermore, since the true parameter is the interior of $D$, by the strong consistency of the parameter estimates, it is easy to see that after some finite time,  the projection operator will become not necessary in the computation process.
	\end{remark}

	\begin{corollary}\label{re7}
		Under conditions of Theorem $\ref{thm2}$, if the excitation condition $(\ref{re})$ is strengthened to  $k=O(\lambda_{\min}(k))$, then the convergence rate of the TSQN algorithm can be improved to the following   iterated logarithm rate:
		\begin{equation}\label{rate}
			\|\tilde{\theta}_{k+1}\|^{2}=O\left(\frac{\log \log k}{k}\right), \;\;a.s.
		\end{equation}
	\end{corollary}	

		We note that the convergence rate $(\ref{rate})$ is known to be the best rate of convergence of the classical least squares in the linear case. It is worth noting that we do not know how to establish such a best possible convergence rate for the preliminary estimate given in the first step.

	Given the parameter estimate $\hat{\theta}_{k}$ by the above TSQN algorithm, one can define an adaptive predictor for the output observation as follows: $$\hat{y}_{k+1}=G_{k}(\phi_{k}^{\top}\hat{\theta}_{k}).$$  Usually, the difference between the best predictor and the adaptive predictor along the sample path, can be measured by  the regret defined  as follows:
	\begin{equation}
		R_{k}=(\mathbb{E}_{k}\left[y_{k+1}\right]-\hat{y}_{k+1})^{2}.
	\end{equation}
	The following theorem gives an asymptotic result for the regret $R_{k}$. 
	
	\begin{theorem} \label{thm1}
		Let Assumptions $\ref{assum1}$-$\ref{assum4}$ hold. Then the sample paths of the accumulated regrets will have the following upper bound:	
		\begin{equation}\label{r}
			\sum_{i=1}^{k}R_{i}=O(\log \lambda_{\max}(k)),\;\;a.s.
		\end{equation}			
	\end{theorem}
	
	The convergence of the accumulated regrets in Theorem $\ref{thm1}$ does not require any excitation condition to hold, and thus can be easily applied to closed-loop control systems. We remark that the order $\log k$ is known to be the best possible order in the linear case (see \cite{lai1986}). 	
	
	\subsection{Asymptotic normality}\label{sub2}
	
	In this subsection, we study the asymptotic distribution properties of the estimation under a general non-PE condition and show that our algorithm is asymptotically efficient in some typical cases. 
	
	For this, let us now take the regulation factor sequence $\{\mu_{k}\}$ in the second step of the TSQN algorithm as
	\begin{equation}\label{alpha}
		\mu_{k}=\sigma_{k}(\phi_{k}^{\top}\hat{\theta}_{k}),
	\end{equation}	
	and the function $\sigma_{k}(\cdot)$ is defined by 
	\begin{equation}\label{sig}
		\sigma_{k}(x)=\mathbb{E}_{k}\left[|S_{k}(x+e_{k+1})-G_{k}(x)|^{2}\right],\;\;a.s.	
	\end{equation}
	Under Assumptions $\ref{assum1}$-$\ref{assum4}$, it is not difficult to obtain that the function $\sigma_{k}(\cdot)$ is Lipschitz continuous and has the following properties: 
	\begin{equation}\label{ssss}
		\sup _{k\geq 0}\mu_{k}\leq\sup _{|x| \leq M_{k}, k\geq 0}|\sigma_{k}(x)|<\infty, \;a.s.
	\end{equation}
	\begin{equation}\label{muu}
		\inf _{k\geq 0}\mu_{k}\geq\inf _{|x| \leq M_{k}, k\geq 0}|\sigma_{k}(x)|>0, \;a.s. 
	\end{equation}
	The proof of $(\ref{ssss})$-$(\ref{muu})$ are provided in Appendix $\ref{A}$. 
	
	We are now in a position to present a theorem on asymptotic normality of the parameter estimate $\hat{\theta}_{k}$ under a general non-PE condition.
	\begin{theorem}\label{thm3}
		Let the Assumptions $\ref{assum1}$-$\ref{assum4}$ be satisfied. Assume that $\{\phi_{k},\;k\geq 0\}$ satisfies as $k\rightarrow \infty$,
		\begin{equation}\label{phi}
			\frac{\log k}{\sqrt{\lambda_{\min}(k)}}=o(1),\;\;a.s.,
		\end{equation}	
		where 
		$\lambda_{\min}(k)$ is the same as that in Theorem $\ref{thm2}$. Besides, assume  that for each $k\geq 0$, there exists a non-random positive definite matrix $\Delta_{k}$ such that as $k\rightarrow \infty$,
		\begin{equation}\label{R1}
			\Delta_{k+1}^{-1}Q_{k+1}^{-\frac{1}{2}} \mathop{\rightarrow}\limits^{p} I,
		\end{equation}	
		where $Q_{k+1}^{-1}=\sum\limits_{i=0}^{k}\frac{(G_{i}'(\phi_{i}^{\top}\theta))^{2}}{\sigma_{i}(\phi_{i}^{\top}\theta)}\phi_{i}\phi_{i}^{\top}+P_{0}^{-1}$,  and $``\mathop{\rightarrow}\limits^{p}"$ means the convergence in probability. Then the estimate $\hat{\theta}_{k}$ given by the TSQN algorithm has the following asymptotically normal property as $k\rightarrow \infty$:
		\begin{equation}\label{theta}
			Q_{k+1}^{-\frac{1}{2}}\tilde{\theta}_{k+1}\mathop{\rightarrow}\limits^{d}\; N(0,I),
		\end{equation}
		where $\tilde{\theta}_{k+1}=\theta-\hat{\theta}_{k+1}$, and $``\mathop{\rightarrow}\limits^{d}"$ means the convergence in distribution.
	\end{theorem}

	\begin{remark}
		Notice that if $\{\phi_{k}\}$ is a deterministic sequence, then $\Delta_{k}$ can be simply chosen as $Q_{k}^{-\frac{1}{2}}$. Moreover, if $\{\frac{(G_{k}'(\phi_{k}^{\top}\theta))^{2}}{\sigma_{k}(\phi_{k}^{\top}\theta)}\phi_{k}\phi_{k}^{\top}\}$ is a random sequence with either ergodic property or $\phi$-mixing property, then under some mild regularity conditions,  it is not difficult to show that $\Delta_{k+1}$ can be taken as $\{\sum_{i=1}^{k}\mathbb{E}[\frac{(G'_{i}(\phi_{i}^{\top}\theta))^{2}}{\sigma_{i}(\phi_{i}^{\top}\theta)}\phi_{i}\phi_{i}^{\top}] \}^{\frac{1}{2}}$. 
	\end{remark}
	
	\begin{remark}
		If we take $l_{k}=-\infty$ and $u_{k}=\infty$, then the system  $(\ref{eq1})$-$(\ref{eq2})$ will degenerate to the standard linear stochastic regression model, and the matrix $Q_{k+1}$ will be the Fisher information matrix provided that $\{\phi_{k}\}$ is a deterministic sequence and $\{e_{k}\}$ is independent with Gaussian distribution. Thus, in this case, our TSQN algorithm is asymptotically efficient \cite{sm2003}. Moreover, the next corollary shows that in the typical binary-valued observation case, our algorithm is also asymptotically efficient.
	\end{remark}

	\begin{corollary}
		Let the conditions of Theorem $\ref{thm3}$ hold and $L_{k}=l_{k}=u_{k}$, where the nonlinear models degenerate to linear regression models with binary-valued observations. If $\{\phi_{k}\}$  is a deterministic sequence and $\{e_{k}\}$ is an independent sequence, then our TSQN algorithm has the following asymptotically normal property  as $k\rightarrow \infty$:
\begin{equation}
			I_{k}^{\frac{1}{2}}\tilde{\theta}_{k}\mathop{\rightarrow}^{d}\; N(0,I),
\end{equation} 		
where $I_{k}$ is the Fisher information matrix given data $\{(y_{i+1}, \phi_{i}),\;\; 0\leq i \leq k\}$. 
	\end{corollary}

	In fact, the fisher information matrix $I_{k}$ in this case can be 	calculated as follows:
	\begin{equation}
		\begin{aligned}
			I_{k+1}=-\mathbb{E}[\frac{\partial^{2}\log p(y_{1}, y_{2}, \cdots, y_{k+1})}{\partial \theta^{2}}]=Q_{k+1}^{-1}.
		\end{aligned}
	\end{equation}	

In the next Remark, we provide a concrete construction of the asymptotic confidence ellipsoid for $\theta$ by Theorem \ref{thm3}. 	
	
	\begin{remark}	\label{re4}
		\textbf{Asymptotic confidence ellipsoid.}   Let $\hat{Q}_{n+1}=(\sum\limits_{i=1}^{n}\frac{(G_{i}'(\phi_{i}^{\top}\hat{\theta}_{i}))^{2}}{\sigma_{i}(\phi_{i}^{\top}\hat{\theta}_{i})} \phi_{i}\phi_{i}^{\top})^{-1}$, from Lemma \ref{lemp} in Section IV and Theorem \ref{thm2}, it is not difficult to obtain that 
\begin{equation}\label{qqh}		
		\|\hat{Q}_{n+1}^{-\frac{1}{2}}Q_{n+1}^{\frac{1}{2}}\|\rightarrow 1, a.s.
\end{equation}		
Thus, by (\ref{qqh}) and Theorem $\ref{thm3}$, we have		
\begin{equation}
\lim_{n\rightarrow \infty} P\{\|\hat{Q}^{-\frac{1}{2}}_{n+1}\tilde{\theta}_{n+1} \|^{2}\leq \mathcal{X}_{m,\alpha}^{2}\}=P\{\mathcal{X}_{m}^{2}\leq \mathcal{X}_{m,\alpha}^{2}\}=1-\alpha,
\end{equation}	
where $\mathcal{X}_{m}^{2}$ is the standard  $\mathcal{X}^{2}-$distribution with degrees of freedom $m$, $m$ is the dimension of the parameter $\theta$, and $\mathcal{X}_{m,\alpha}^{2}$ is the $\alpha-$quantile of $\mathcal{X}_{m}^{2}$. Therefore, the $1-\alpha$ asymptotically correct confidence ellipsoids for $\theta$ take the form $\{\theta:\|\hat{Q}^{-\frac{1}{2}}_{n+1}\tilde{\theta}_{n+1} \|^{2}\leq \mathcal{X}_{m,\alpha}^{2}\}$ ( \cite{dp2014},\cite{sm2003}).	Moreover, since $\|\epsilon_{j}^{\top}\tilde{\theta}_{n+1}\| \leq \|\epsilon_{j}^{\top}\hat{Q}_{n+1}^{\frac{1}{2}}\| \cdot\|\hat{Q}_{n+1}^{-\frac{1}{2}}\tilde{\theta}_{n+1} \|,$ where $\epsilon_{j}$ is the $j^{th}$ column of the identity matrix, we can  get a more detailed  $1-\alpha$ asymptotically correct confidence intervals for the components of the estimation error vector as follows: 
		\begin{equation}
			C^{(j)}=\left(\hat{\theta}_{n+1}^{(j)}-\sqrt{\hat{Q}_{n+1}^{(j)}\mathcal{X}_{m,\alpha}^{2}}\;,\; \hat{\theta}_{n+1}^{(j)}+\sqrt{\hat{Q}_{n+1}^{(j)}\mathcal{X}_{m,\alpha}^{2}}\right)
		\end{equation}
		where $\hat{Q}_{n+1}^{(j)}$ is the  $j^{th}$ diagonal element of $\hat{Q}_{n+1}$.
	\end{remark}	

	\subsection{Non-asymptotic analysis}\label{sub3}		
	Though an asymptotic bound can be given based on Theorem $\ref{thm3}$ as discussed in Remark ${\ref{re4}}$ which, however, requires that the number of data samples is sufficiently large. As a result,  the asymptotic bound is hard to apply in real situations where one only has a finite number of data samples.  In this subsection, we will provide some upper bounds for the estimation errors with high probability when the number of data samples is given and finite.  	
	
	\subsubsection{Lyapunov function-based confidence interval}	
	In this subsection, we give a Lyapunov function-based confidence interval based on the analysis of a Lyapunov function used in the proof of Lemma $\ref{lem6}$. For convenience, let the initial values in the TSQN algorithm satisfy $\log|P_{0}^{-1}|=\log|\bar{P}_{0}^{-1}|>1$ and introduce the following notations to be used throughout the sequel:
	\begin{equation}\label{star}
		w_{k+1}=y_{k+1}-G_{k}(\phi_{k}^{\top}\theta).
	\end{equation}
	\begin{theorem}\label{thm4}
		Under Assumptions $\ref{assum1}$-$\ref{assum4}$, assume that $\{w_{k+1}^{2}\}$ is an $\mathcal{L}_{2}$ sequence. Then  for any given $N\geq 1$ and any $0<\alpha<\frac{1}{2}$, $1\leq j \leq m$, each component $\tilde{\theta}_{N+1}^{(j)}$ of $\tilde{\theta}_{N+1}$ satisfies the following inequality with probability at least $1-2\alpha$:
		\begin{equation}\label{eeee}
			\begin{aligned}
				|\tilde{\theta}_{N+1}^{(j)}|^{2} \leq & P_{N+1}^{(j)}(\sigma_{b}\log |P_{N+1}^{-1}|+\frac{2\Psi C}{\tau\lambda_{N}}+\Gamma+c_{0}), 
			\end{aligned}
		\end{equation}
	and with probability at least $1-2\alpha$, we have
		\begin{equation}\label{regret}
					\begin{aligned}
			R_{N}\leq &2\delta_{0}(\sigma_{b}\log |P_{N+1}^{-1}|+\frac{2\Psi C}{\tau\lambda_{N}}+\Gamma+c_{0}),
					\end{aligned}
		\end{equation}
	where
	\begin{equation}
	\begin{aligned}
		C=&2(\bar{\mu}\sigma_{b}+\bar{\Gamma}+1)^{2+\tau}+[(1+6\gamma\bar{\mu}^{2})(\sigma_{b}+\Gamma+ \frac{\bar{\Gamma}}{\bar{\mu}^{2}}+1)]^{2+\tau},\\
		\Gamma=&V_{0}+\sigma_{b}\log|P_{0}^{-1}|+\frac{\Phi tr(P_{0})\bar{\sigma}_{b}}{2\sigma_{b}}+\frac{18 \sigma_{b}(1-\alpha)}{\alpha},\\
		\bar{\Gamma}=&\bar{V}_{0}+\bar{\mu}\sigma_{b}\log|\bar{P}_{0}^{-1}|+\frac{\bar{\Phi}tr(\bar{P}_{0})\bar{\mu}\bar{\sigma}_{b}}{2\sigma_{b}}+\frac{10\bar{\mu}\sigma_{b}(1-\alpha)}{\alpha},
	\end{aligned}
\end{equation}
	and $P_{n+1}^{(j)}$ is the $j^{th}$ diagonal component of $P_{n+1}$, $\gamma=\sup\limits_{0\leq k \leq N}\frac{\mu_{k}^{-1}\overline{g}_{k}^{2}}{\bar{a}_{k}^{2}\bar{\beta}_{k}^{2}}$, $c_{0}=\frac{3}{2}\mu_{0}^{-1}\overline{g}_{0}^{2}(\phi_{0}^{\top}\tilde{\bar{\theta}}_{0})^{2}$, $\delta_{0}=\sup\limits_{0\leq k \leq N}\{\mu_{k}+\beta_{k}^{2}\phi_{k}^{\top}P_{k}\phi_{k}\}$, $\Phi=\sup\limits_{0\leq k\leq N}\mu_{k}^{-1}\beta_{k}^{2}\|\phi_{k}\|^{2}$, $\bar{\Phi}=\sup\limits_{0\leq k\leq N}\bar{\beta}_{k}^{2}\|\phi_{k}\|^{2}$, $\Psi=\sup\limits_{0\leq k\leq N}\frac{6\mu_{k}^{-1}\rho^{2}\|\phi_{k}\|^{2}}{\bar{\beta}_{k}^{2} max\{a_{k},\bar{a}_{k}\}},$ $\sigma_{b}=\sup\limits_{0\leq k\leq N}\mu_{k}^{-1}\mathbb{E}_{k}\left[w_{k+1}^{2}\right]$,  $\overline{\sigma}_{b}=\sup\limits_{0\leq k\leq N}\mu_{k}^{-2}\mathbb{E}_{k}\left[(w_{k+1}^{2}-\mathbb{E}_{k}\left[w_{k+1}^{2}\right])^{2}\right]$,   $\bar{\mu}=\sup\limits_{0\leq k \leq N}\mu_{k}+1$, $\lambda_{N}=\inf\limits_{0\leq k \leq N}\{\frac{\lambda_{\min}\{\bar{P}_{k}^{-1}\}}{(\log|\bar{P}_{k+1}^{-1}|)^{2+\tau}},\frac{\lambda_{\min}\{P_{k}^{-1}\}}{(\log |P_{k+1}^{-1}|)^{2+\tau}}\}$, $\tau>0$.
	\end{theorem}

	In contrast to Remark $\ref{re4}$ where the number of data samples is sufficiently large, the above Theorem $\ref{thm4}$ can provide a concrete confidence interval for any given finite number of data samples. Next, we provide an alternative confidence interval by using the  Monte Carlo method, which turns out to have some advantages also in the case of finite data samples.		
	
	\subsubsection{Monte Carlo-based confidence interval}
	
	In this subsection, we give a Monte Carlo-based confidence interval by designing a Monte Carlo experiment.
	
	Consider the nonlinear stochastic system defined by $(\ref{eq1})$-$(\ref{eq2})$ and the adaptive nonlinear TSQN algorithm defined by $(\ref{be1})$-$(\ref{be2})$. Suppose that the unknown system parameter $\theta \in D$ is a random vector with uniform distribution $U$, that the system noise $\{e_{i}\}_{i=1}^{n}$ is an i.i.d sequence which is independent of  $\theta$ with distribution $F$,  and that both the saturation functions $\{S_{i}(\cdot)\}_{i=1}^{n}$ and system regressors $\{\phi_{i}\}_{i=1}^{n}$ are deterministic sequence, where $n$ is a given fixed data length. It is easy to see that the vector
	$X=(\theta^{\top}, e_{1}, \cdots,e_{n})$ has a joint distribution $P=U \times F^{n}.
	$	
	
	To construct the Monte Carlo-based confidence interval, let $\{X_{1}, X_{2}, \cdots, X_{K}\}$ be K samples taken from the joint distribution $P$, and generate  the corresponding $n$-dimensional observation set $\{Y_{1}, Y_{2}, \cdots, Y_{K}\}$ by the model $(\ref{eq1})$-$(\ref{eq2})$ together with the given data of regressors $\{\phi_{i}\}_{i=1}^{n}$. Then compute the estimation error $\{\tilde{\theta}_{n,1}^{(j)}, \tilde{\theta}_{n,2}^{(j)}, \cdots, \tilde{\theta}_{n,K}^{(j)}\}$ by the TSQN algorithm for any $j=1,2\cdots,m$. It is easy to see that there is a measurable function $H_{j}: R^{n+m}\rightarrow R$ and a probabilistic distribution function $F^{(j)}$, such that the $j^{th}$ component of the parameter estimation error  $\tilde{\theta}_{n, i}^{(j)}$ can be expressed as follows for any $1\leq i \leq K$ and $j=1,2\cdots,m$:
	\begin{equation}
		\begin{aligned}
			&\tilde{\theta}_{n,i}^{(j)}=H_{j}(X_{i})\sim F^{(j)}.\;
		\end{aligned}
	\end{equation} 		
	Let the empirical distribution function of the generated samples for the $j^{th}$ component be 
	$$F_{K}^{(j)}(x)=\frac{1}{K}\sum_{i=1}^{K}I_{\{H_{j}(X_{i})\leq x\}}, \;\; \forall x \in R,\;\;j=1,\cdots,m.$$
	
	Under the above-mentioned assumptions and notations, we have the following proposition on the confidence interval with finite data length:
	
	\begin{proposition}\label{po1}	
		For any positive $\alpha$ and $t$ with $\alpha+t<1$, and any $j=1,2,\cdots,m$, the $j^{th}$ component of the estimation error $\tilde{\theta}_{n}^{(j)}$ generated by the TSQN algorithm  belongs to the following confidence interval with probability  at least $1-\alpha-t$: 
		\begin{equation}
			\begin{aligned}
				\tilde{\theta}_{n}^{(j)} \in [z_{K}^{(j)}(\frac{\alpha}{2}-\sqrt{\frac{\ln 2-\ln t}{2K}}), \;\;z_{K}^{(j)}(1-\frac{\alpha}{2}+\sqrt{\frac{\ln 2-\ln t}{2K}})].
			\end{aligned}
		\end{equation}
		where $z^{(j)}_{K}(\alpha)$ is the $\alpha$ quantiles of the distribution $F_{K}^{(j)}$. 
	\end{proposition}	
	
	It is obvious that the confidence interval of the estimation error given in Proposition $\ref{po1}$ will decrease asymptotically, as the number of random samplings $K$  used in Proposition $\ref{po1}$ increases. We remark that there are at least two advantages of Proposition $\ref{po1}$, one is that it is applicable to the case where the data length $n$ is given and finite, and another is that the confidence interval may be better than that given in Theorem $\ref{thm4}$ in some applications.

	\section{Proofs of the main results}\label{ss4}
			
	\subsection{Proof of Theorem $\ref{thm2}$ and Theorem $\ref{thm1}$.}
		For convenience, we denote
	\begin{equation}\label{24}
		\bar{\psi}_{k}=G_{k}(\phi_{k}^{\top}\theta)-G_{k}(\phi_{k}^{\top}\bar{\theta}_{k}),
	\end{equation}
	\begin{equation}\label{psi}
		\psi_{k}=G_{k}(\phi_{k}^{\top}\theta)-G_{k}(\phi_{k}^{\top}\hat{\theta}_{k}).
	\end{equation}
To prove Theorem $\ref{thm2}$ and Theorem $\ref{thm1}$,  we need to establish the following lemma first.	
	
	\begin{lemma}\label{lem6}
		Let Assumptions $\ref{assum1}$-$\ref{assum4}$ be satisfied. Then the parameter estimate $\hat{\theta}_{k}$ given by  TSQN Algorithm  has the following property as  $k \rightarrow \infty$:
		\begin{equation}\label{22}
			\begin{aligned}
				\tilde{\theta}_{k+1}^{\top}P_{k+1}^{-1}\tilde{\theta}_{k+1}+\sum_{i=0}^{k}a_{i}\psi_{i}^{2}	= O\left(\log \lambda_{\max}(k)\right).
			\end{aligned}
		\end{equation}
		where $\tilde{\theta}_{k}$ is defined  as $\theta-\hat{\theta}_{k}$, $\psi_{k}$ is defined as in $(\ref{psi})$
	\end{lemma}	
	\begin{proof}Following the analysis ideas of the classical least-squares for linear stochastic regression models (see e.g., \cite{mj1978}, \cite{lw1982}, \cite{g1995}), we consider the following stochastic Lyapunov function:
		\begin{equation}\label{v}
			V_{k+1}=\tilde{\theta}_{k+1}^{\top}P_{k+1}^{-1}\tilde{\theta}_{k+1}.
		\end{equation}	
		By  $(\ref{be2})$,  we know that
		\begin{equation}\label{P-1}
			P_{k+1}^{-1}=P_{k}^{-1}+\mu_{k}^{-1}\beta_{k}^{2}\phi_{k}\phi_{k}^{\top}.
		\end{equation}
		Hence, multiplying $a_{k}\phi_{k}^{\top}P_{k}$ from the left hand side and noticing the definition of $a_{k}$, we know that
		\begin{equation}\label{aalpha}
			\begin{aligned}
				&a_{k}\phi_{k}^{\top}P_{k}P_{k+1}^{-1}\\
				=&a_{k}\phi_{k}^{\top}(I+\mu_{k}^{-1}\beta_{k}^{2}P_{k}\phi_{k}\phi_{k}^{\top})=\mu_{k}^{-1}\phi_{k}^{\top}.
			\end{aligned}
		\end{equation}
		Also by $(\ref{psi})$ and the definition of $\beta_{k}$ in $(\ref{be2})$,  we know that
		\begin{equation}
			\begin{aligned}
				&G_{k}(\phi_{k}^{\top}\theta)-G_{k}(\phi_{k}^{\top}\bar{\theta}_{k})\\
				=&G_{k}(\phi_{k}^{\top}\theta)-G_{k}(\phi_{k}^{\top}\hat{\theta}_{k})+G_{k}(\phi_{k}^{\top}\hat{\theta}_{k})-G_{k}(\phi_{k}^{\top}\bar{\theta}_{k})\\
				=&\psi_{k}+\beta_{k}\phi_{k}^{\top}(\hat{\theta}_{k}-\bar{\theta}_{k})
				=\psi_{k}+\beta_{k}\phi_{k}^{\top}(\theta-\bar{\theta}_{k}-\tilde{\theta}_{k})
			\end{aligned}
		\end{equation}
		Hence,
		\begin{equation}\label{barpsi}
			\psi_{k}-\beta_{k}\phi_{k}^{\top}\tilde{\theta}_{k}=\bar{\psi}_{k}-\beta_{k}\phi_{k}^{\top}\tilde{\bar{\theta}}_{k}.
		\end{equation}
		Moreover, by Lemma $\ref{lem1}$ in Appendix \ref{ii}, $(\ref{v})$, $(\ref{aalpha})$,  and $(\ref{barpsi})$, we know that
		\begin{equation}\label{61}
			\begin{aligned}
				V_{k+1}\leq&[\tilde{\theta}_{k}-a_{k}\beta_{k}P_{k}\phi_{k}(\psi_{k}+w_{k+1})]^{\top}P_{k+1}^{-1}\cdot\\
				&[\tilde{\theta}_{k}-a_{k}\beta_{k}P_{k}\phi_{k}(\psi_{k}+w_{k+1})]\\
				=&V_{k}+\mu_{k}^{-1}\beta_{k}^{2}(\phi_{k}^{\top}\tilde{\theta}_{k})^{2}-2a_{k}\beta_{k}\phi_{k}^{\top}P_{k}P_{k+1}^{-1}\tilde{\theta}_{k}\psi_{k}\\
				&+a_{k}^{2}\beta_{k}^{2}\phi_{k}^{\top}P_{k}P_{k+1}^{-1}P_{k}\phi_{k}\psi_{k}^{2}-2a_{k}\beta_{k}\phi_{k}^{\top}P_{k}P_{k+1}^{-1}\tilde{\theta}_{k}w_{k+1}\\
				&+2a_{k}^{2}\beta_{k}^{2}\phi_{k}^{\top}P_{k}P_{k+1}^{-1}P_{k}\phi_{k}\psi_{k}w_{k+1}\\
				&+\beta_{k}^{2}a_{k}^{2}\phi_{k}^{\top}P_{k}P_{k+1}^{-1}P_{k}\phi_{k}w_{k+1}^{2}\\
				=&V_{k}-\mu_{k}^{-1}\psi_{k}^{2}+\mu_{k}^{-1}(\psi_{k}-\beta_{k}\phi_{k}^{\top}\tilde{\theta}_{k})^{2}\\
				&+\mu_{k}^{-1}a_{k}\beta_{k}^{2}\phi_{k}^{\top}P_{k}\phi_{k}\psi_{k}^{2}+2\mu_{k}^{-1}(\psi_{k}-\beta_{k}\phi_{k}^{\top}\tilde{\theta}_{k})w_{k+1}\\
				&-2\mu_{k}^{-1}\psi_{k}w_{k+1}+2\mu_{k}^{-1}a_{k}\beta_{k}^{2}\phi_{k}^{\top}P_{k}\phi_{k}\psi_{k}w_{k+1}\\ &+\mu_{k}^{-1}a_{k}\beta_{k}^{2}\phi_{k}^{\top}P_{k}\phi_{k}\psi_{k}w_{k+1}^{2}\\
				=& V_{k}-a_{k}\psi_{k}^{2}+\mu_{k}^{-1}(\bar{\psi}_{k}-\beta_{k}\phi_{k}^{\top}\tilde{\bar{\theta}}_{k})^{2}-2a_{k}\psi_{k}w_{k+1}\\		
				+&2\mu_{k}^{-1}(\psi_{k}-\beta_{k}\phi_{k}^{\top}\tilde{\theta}_{k})w_{k+1}
				+\mu_{k}^{-1}a_{k}\beta_{k}^{2}\phi_{k}^{\top}P_{k}\phi_{k}w_{k+1}^{2},\\
			\end{aligned}
		\end{equation}
		where $w_{k+1}$ is defined in $(\ref{star})$. Summing up both sides of $(\ref{61})$ from $0$ to $n$ and using $(\ref{barpsi})$, we have
		\begin{equation}\label{VV}
			\begin{aligned}
				V_{n+1}\leq &V_{0}-\sum_{k=0}^{n}a_{k}\psi_{k}^{2}+\sum_{k=0}^{n}\mu_{k}^{-1}(\bar{\psi}_{k}-\beta_{k}\phi_{k}^{\top}\tilde{\bar{\theta}}_{k})^{2}\\
				&-\sum_{k=0}^{n}2a_{k}\psi_{k}w_{k+1}
				+\sum_{k=0}^{n}2\mu_{k}^{-1}(\bar{\psi}_{k}-\beta_{k}\phi_{k}^{\top}\tilde{\bar{\theta}}_{k})w_{k+1}\\
				&+\sum_{k=0}^{n}\mu_{k}^{-1}a_{k}\beta_{k}^{2}\phi_{k}^{\top}P_{k}\phi_{k}w_{k+1}^{2},\;\;a.s.
			\end{aligned}
		\end{equation}			 
		
		We now analyze the RHS of $(\ref{VV})$ term by term.

		First, by $(\ref{wwww})$ in Appendix \ref{A}, we have  
		\begin{equation}\label{www}
		\sup\limits_{|x|\leq M_{k}, k\geq 0}\mathbb{E}_{k}\left[|S_{k}(x+e_{k+1})|^{2+\eta}\right]<\infty, \;\;a.s.
		\end{equation}
	        Since $\phi_{k}^{\top}\theta$ is $\mathcal{F}_{k}-$measurable and $|\phi_{k}^{\top}\theta|\leq M_{k},\;a.s.$ by $(\ref{Mk})$, we can easily have $\sup\limits_{k\geq 0}\mathbb{E}_{k}\left[|w_{k+1}|^{2+\eta}\right]<\infty,\; a.s.$ 
	         Thus, by  Lemma $\ref{lem2}$ in  Appendix \ref{ii}, we know that
		\begin{equation}\label{e1}
			\begin{aligned}
				\sum_{k=1}^{n}2a_{k}\psi_{k}w_{k+1}
				=o(\sum_{k=1}^{n}a_{k}\psi_{k}^{2})+O(1),\;\;\;a.s.,\;\;\forall \gamma >0,
			\end{aligned}
		\end{equation} 
		where we have used the fact that $\sup\limits_{k\geq 0}\{a_{k}\}\leq\sup\limits_{k\geq 0}\{\mu_{k}^{-1}\}<\infty,a.s$. Similarly, we have
		\begin{equation}\label{ee2}
			\begin{aligned}
				&\sum_{k=1}^{n}2\mu_{k}^{-1}[\bar{\psi}_{k}-\beta_{k}\phi_{k}^{\top}(\theta-\bar{\theta}_{k})]w_{k+1}\\
				=&o(\sum_{k=1}^{n}\mu_{k}^{-1}[\bar{\psi}_{k}-\beta_{k}\phi_{k}^{\top}(\theta-\bar{\theta}_{k})]^{2})+O(1),\;\;\;\;\;a.s.	
			\end{aligned}
		\end{equation}
		
		For the last term on the RHS  of $(\ref{VV})$, 
		let  us take $X_{k}=\beta_{k}\phi_{k}$ in Lemma $\ref{lem3}$ in Appendix \ref{ii}, we get
		\begin{equation}\label{35}
			\sum_{k=0}^{n}a_{k}\beta_{k}^{2}\phi_{k}^{\top}P_{k}\phi_{k}=O(\log \lambda_{\max}(n) ),\;\;\;a.s.
		\end{equation}
		Moreover, from Lyapunov inequality,  we have for any $\delta\in (2, \min(\eta,4))$
		\begin{equation}\label{final}
			\sup_{k\geq 0} \mathbb{E}_{k}\left[\left|w_{k+1}^{2}-\mathbb{E}_{k}\left[w_{k+1}^{2}\right]\right|^{\frac{\delta}{2}}\right] < \infty, \;a.s.
		\end{equation}
		Denote $\Lambda_{n} =(\sum_{k=0}^{n}\left(a_{k}\beta_{k}^{2} \phi_{k}^{\top} P_{k} \phi_{k}\right)^{\frac{\delta}{2}})^{\frac{2}{\delta}}$, by Lemma $\ref{lem2}$ in Appendix \ref{ii} with $\alpha = \frac{\delta}{2}$, we get
		\begin{equation}\label{36}
			\begin{aligned}
				&\sum_{k=0}^{n}\mu_{k}^{-1}a_{k}\beta_{k}^{2}\phi_{k}^{\top}P_{k}\phi_{k}(w_{k+1}^{2}-	\mathbb{E}_{k}\left[w_{k+1}^{2}\right])\\
				=&O\left(\Lambda_{n} \log^{\frac{1}{2}+\gamma} (\Lambda_{n}^{2}+e)\right)\\
				=&o(\log \lambda_{\max}(n) )+O(1), \quad \text { a.s. } \forall \gamma>0.
			\end{aligned}
		\end{equation}
		Hence, from $(\ref{35})$ and $(\ref{36})$
		\begin{equation}\label{39}
			\begin{aligned}
				& \sum_{k=0}^{n}\mu_{k}^{-1} a_{k}\beta_{k}^{2} \phi_{k}^{\top} P_{k} \phi_{k} w_{k+1}^{2} \\
				\leq& \sum_{k=0}^{n}\mu_{k}^{-1}a_{k}\beta_{k}^{2}\phi_{k}^{\top}P_{k}\phi_{k}\left(w_{k+1}^{2}-	\mathbb{E}_{k}\left[w_{k+1}^{2}\right]\right)\\
				&+ \sup_{k\geq 0}\mathbb{E}_{k}\left[w_{k+1}^{2}\right]\left( \sum_{k=0}^{n}\mu_{k}^{-1}a_{k}\beta_{k}^{2}\phi_{k}^{\top}P_{k}\phi_{k}\right) \\
				=& O(\log \lambda_{\max}(n) )\quad \text { a.s. }
			\end{aligned}
		\end{equation}
		The analysis of the third term on the RHS of $(\ref{VV})$ is a key feature for analyzing the TQSN algorithm since in the single-step algorithm this term does not exist by the construction of the scalar adaptation gain. Now let
		\begin{equation}
			\zeta_{k}=\frac{\bar{\psi}_{k}}{\phi_{k}^{\top}\tilde{\bar{\theta}}_{k}}I_{\{\phi_{k}^{\top}\tilde{\bar{\theta}}_{k}\not=0\}}+\underline{g}_{k}I_{\{\phi_{k}^{\top}\tilde{\bar{\theta}}_{k}=0\}},
		\end{equation}
by $(\ref{g}$), we then have
\begin{equation}
0<\zeta_{k}\leq\overline{g}_{k},\;\;0<\beta_{k}\leq\overline{g}_{k},\;\;a.s.
\end{equation}
Hence, we can obtain that
		\begin{equation}\label{fr}
				\begin{aligned}
			&\sum_{k=1}^{n}\mu_{k}^{-1}(\bar{\psi}_{k}-\beta_{k}\phi_{k}^{\top}\tilde{\bar{\theta}}_{k})^{2}\\
			=&\sum_{k=1}^{n}\mu_{k}^{-1}(\zeta_{k}-\beta_{k})^{2}(\phi_{k}^{\top}\tilde{\bar{\theta}}_{k})^{2}=O(\sum_{k=1}^{n}(\phi_{k}^{\top}\tilde{\bar{\theta}}_{k})^{2}),
		\end{aligned}
		\end{equation}
		where we have used the fact that $|\zeta_{k}-\beta_{k}|\leq \sup\limits_{k\geq 0}\overline{g}_{k}<\infty$. We now prove 
		\begin{equation}\label{ffr}
			\sum_{k=1}^{n}(\phi_{k}^{\top}\tilde{\bar{\theta}}_{k})^{2}=O(\log \lambda_{\max}(n)).
		\end{equation}
		For this, we consider the Lyapunov function 
		\begin{equation}\label{558}
			\bar{V}_{k+1}=\tilde{\bar{\theta}}_{k+1}^{\top}\bar{P}_{k+1}^{-1}\tilde{\bar{\theta}}_{k+1},
		\end{equation}
Similarly to $(\ref{61})$, we have the following property:
	\begin{equation}\label{su}
		\begin{aligned}
			\bar{V}_{n+1}\leq & \bar{V}_{0}-\sum_{k=0}^{n}(\bar{\beta}_{k}\tilde{\bar{\theta}}_{k}^{\top}\phi_{k}\bar{\psi}_{k}-\bar{a}_{k}\bar{\beta}_{k}^{2}\phi_{k}^{\top}\bar{P}_{k}\phi_{k}\bar{\psi}_{k}^{2})\\
			&+\sum_{k=0}^{n}\bar{a}_{k}\bar{\beta}_{k}^{2}\phi_{k}^{\top}\bar{P}_{k}\phi_{k}\mathbb{E}_{k}\left[w_{k+1}^{2}\right]\\
			&-2\sum_{k=0}^{n}(\bar{\beta}_{k}\phi_{k}^{\top}\tilde{\bar{\theta}}_{k}-\bar{a}_{k}\bar{\beta}_{k}^{2}\bar{\psi}_{k}\phi_{k}^{\top}\bar{P}_{k}\phi_{k})w_{k+1}\\
			&+\sum_{k=0}^{n}\bar{a}_{k}\bar{\beta}_{k}^{2}\phi_{k}^{\top}\bar{P}_{k}\phi_{k}(w_{k+1}^{2}-\mathbb{E}_{k}\left[w_{k+1}^{2}\right]),\;\;a.s.
		\end{aligned}
	\end{equation} 
By the definition of $\bar{\beta}_{k}$,  we have $|\bar{\beta}_{k}\phi_{k}^{\top}\tilde{\bar{\theta}}_{k}-\bar{a}_{k}\bar{\beta}_{k}^{2}\bar{\psi}_{k}\phi_{k}^{\top}\bar{P}_{k}\phi_{k}|\leq |\bar{\psi}_{k}|$. Besides, following the similar analysis for the noise term as  in $(\ref{e1})$-$(\ref{39})$, we will have
\begin{equation}\label{bb}
		\begin{aligned}
&\bar{V}_{n+1}+\sum_{k=0}^{n}(\bar{\beta}_{k}\tilde{\bar{\theta}}_{k}^{\top}\phi_{k}\bar{\psi}_{k}-\bar{a}_{k}\bar{\beta}_{k}^{2}\phi_{k}^{\top}\bar{P}_{k}\phi_{k}\bar{\psi}_{k}^{2})\\
=&O(\log \lambda_{\max}(n)),\;a.s.
	\end{aligned}
\end{equation}
Also by the definition of $\bar{\beta}_{k}$ in $(\ref{be1})$ and $\bar{\psi}_{k}$ in $(\ref{24})$, we have
\begin{equation}
	\bar{\psi}_{k}^{2}\geq\underline{g}_{k}^{2}(\phi_{k}^{\top}\tilde{\bar{\theta}}_{k})^{2}\geq \bar{\beta}_{k}^{2}(\phi_{k}^{\top}\tilde{\bar{\theta}}_{k})^{2},\;a.s.
\end{equation}
and 
	\begin{equation}\label{com}
			\begin{aligned}
	&(\bar{\beta}_{k}\tilde{\bar{\theta}}_{k}^{\top}\phi_{k}\bar{\psi}_{k}-\bar{a}_{k}\bar{\beta}_{k}^{2}\phi_{k}^{\top}\bar{P}_{k}\phi_{k}\bar{\psi}_{k}^{2})\\
	\geq & \frac{1}{2}\bar{a}_{k}\bar{\beta}_{k}\phi_{k}^{\top}\tilde{\bar{\theta}}_{k}\bar{\psi}_{k}	\geq  \frac{1}{2}\bar{a}_{k}\bar{\beta}_{k}^{2}(\phi_{k}^{\top}\tilde{\bar{\theta}}_{k})^{2},\;a.s.
		\end{aligned}
\end{equation}
Moreover, Since $\{\bar{\beta}_{k}\}$ and $\{\phi_{k}\}$ are bounded, we obtain that 
\begin{equation}\label{aa}
	\begin{aligned}
		\inf_{k \geq 0}\{\bar{a}_{k}\} &\geq \inf_{k \geq 0}\{\frac{1}{1+\bar{\beta}_{k}^{2}\phi_{k}^{\top}\bar{P}_{0}\phi_{k}}\}>0, \;\;a.s. .
	\end{aligned}
\end{equation} 
Note that $\{\bar{\beta}_{k}\}$ has  a  positive lower bounded almost surely, $(\ref{ffr})$ can be obtained  by $(\ref{bb})$  $(\ref{com})$ and $(\ref{aa})$ .
	
		Finally, combining $(\ref{VV})$, $(\ref{e1})$, $(\ref{ee2})$, $(\ref{39})$, $(\ref{fr})$ and $(\ref{ffr})$, we get the desired result $(\ref{22})$.
	\end{proof}		

\noindent\hspace{2em}{\itshape Proof of Theorem $\ref{thm2}$ and Theorem $\ref{thm1}$: }	
For the proof of Theorem $\ref{thm2}$, by $(\ref{v})$ and $(\ref{P-1})$, we have
\begin{equation}
\begin{aligned}
V_{n+1}
 \geq \epsilon_{0}\lambda_{\min}\{\sum_{k=0}^{n}\phi_{k}\phi_{k}^{\top}+P_{0}^{-1}\}\|\tilde{\theta}_{n+1}\|^{2}, \;\;a.s.,
 \end{aligned}
 \end{equation}
 where $\epsilon_{0}=\min\{1, \inf\limits_{k\geq 0}(\mu_{k}^{-1}\underline{g}_{k}^{2})\}$, which is positive by $(\ref{zh})$. Hence,
 Theorem $\ref{thm2}$ follows immediately from Lemma $\ref{lem6}$ and $\epsilon_{0}>0$.

 For the proof of Theorem $\ref{thm1}$,  from the definition of $\psi_{k}$ in $(\ref{psi})$, we have 
	\begin{equation}\label{pssi}
	\psi_{k}^{2}\geq \underline{g}_{k}^{2}(\phi_{k}^{\top}\tilde{\theta}_{k})^{2},
	\end{equation}
	 where $\inf\limits_{k\geq 0}\underline{g}_{k}>0$ by $(\ref{zh})$. Besides, Since $\{\mu_{k}\}$, $\{\beta_{k}\}$ and $\{\phi_{k}\}$ are bounded, we obtain that 
	\begin{equation}
			\begin{aligned}
	 \inf_{k \geq 0}a_{k}=\inf_{k \geq 0}\frac{1}{\mu_{k}+\beta_{k}^{2}\phi_{k}^{\top}P_{k}\phi_{k}}>0, \;\;a.s. 
	\end{aligned}
	 	\end{equation}
	 		Thus  Theorem $\ref{thm1}$ also follows from Lemma $\ref{lem6}$.
	\hspace*{\fill}~\QED\par\endtrivlist\unskip
	\subsection{ Proof of Theorem $\ref{thm3}$.}
	
	\begin{lemma}\label{lemp}
		Let $\{X_{i},  i=1,2,\cdots\}$ be a sequence of random variables in $\mathbb{R}^{p} (p\geq 1)$, and $\{a_{i},  i=1,2,\cdots\}$ be a sequence of random variables in $\mathbb{R}$. Also, let $A_{n}=\sum_{i=1}^{n}X_{i}X_{i}^{\top}+X_{0}$. If
		\begin{equation}
			\begin{aligned}
				\lambda_{max}\{A_{n}^{-\frac{1}{2}}\}\rightarrow 0,\;\;\;
				a_{i}\rightarrow 1,\;\;\;a.s.
			\end{aligned}
		\end{equation} 
		then
		\begin{equation}
			A_{n}^{-\frac{1}{2}}(\sum_{i=1}^{n}X_{i}X_{i}^{\top}a_{i}^{2}+X_{0})A_{n}^{-\frac{1}{2}}\rightarrow I,\;\;\;a.s.		\end{equation}	 		
	\end{lemma}	
	\begin{proof}	
		For every $m\geq 1$, $n\geq m$, we obtain that
		\begin{equation}\label{xx}
			\begin{aligned}
				&\|A_{n}^{-\frac{1}{2}}(\sum_{i=1}^{n}X_{i}X_{i}^{\top}a_{i}^{2}+X_{0})A_{n}^{-\frac{1}{2}}-I\|	\\
				= & \|A_{n}^{-\frac{1}{2}}(\sum_{i=1}^{n}X_{i}X_{i}^{\top}(a_{i}^{2}-1))	A_{n}^{-\frac{1}{2}}\|	\\
				\leq & \|A_{n}^{-\frac{1}{2}}(\sum_{i=1}^{m}X_{i}X_{i}^{\top}|a_{i}^{2}-1|)A_{n}^{-\frac{1}{2}}\|\\
				&+\|A_{n}^{-\frac{1}{2}}(\sum_{i=m}^{n}X_{i}X_{i}^{\top})A_{n}^{-\frac{1}{2}}\|\cdot \sup_{m+1 \leq i \leq n}|a_{i}^{2}-1|,\;\;a.s.,
			\end{aligned}
		\end{equation}
		Let $n\rightarrow \infty$ and $m\rightarrow \infty$, since $a_{i}\rightarrow 1$ almost surely, we have the RHS of $(\ref{xx})$ converges to 0 almost surely. Lemma $\ref{lemp}$ thus be proven.
	\end{proof}	
	
	\begin{lemma}\label{lempp}
		Under Assumptions $\ref{assum1}$-$\ref{assum4}$ and condition $(\ref{re})$, let
				\begin{equation}\label{sss} s_{k}=\hat{\theta}_{k}+a_{k}\beta_{k}P_{k}\phi_{k}[y_{k+1}-G_{k}(\phi_{k}^{\top}\hat{\theta}_{k})],		\end{equation} 	
		and let $\mathcal{A}_{k}=\{s_{k}\not\in D\}$. Then
		\begin{equation}
			P\{\omega: \omega \in \mathcal{A}_{k}, i.o.\}=0,
		\end{equation}
	where $i.o.$ means the related event occurs infinitely often.
	\end{lemma}	
\begin{proof}
From Assumption $\ref{assum1}$, there exists a ball centered at $\theta$ with radius $r>0$, such that $B(\theta, r) \subseteq   int(D)$, the interior of the set $D$.  Since by Remark $\ref{re77}$ the estimate $\hat{\theta}_{k}$ is strongly consistent,  there exists a random integer  $N$ such that for any $k\geq N$,
	\begin{equation}\label{bbb}
	 \hat{\theta}_{k} \in B(\theta, r)\subseteq   int(D),\;\;a.s.
	 \end{equation}

	Now,  let 
	\begin{equation}
	\mathcal{H}=\{\omega: \lim\limits_{k\rightarrow\infty}\frac{\log \lambda_{\max}(k)}{\lambda_{\min}(k)}=0\}.
	\end{equation}
	For any $\omega_{0} \in \mathcal{H}$, we prove that $\omega_{0}\not \in \{\omega: \omega \in \mathcal{A}_{k}, i.o.\}$ by contradiction.  If $\omega_{0}\in \{\omega: \omega \in \mathcal{A}_{k}, i.o.\}$, then there exists a $k_{0}>N$, such that $s_{k_{0}}\not\in D$. Thus $\hat{\theta}_{k_{0}+1}=\Pi_{P_{k_{0}+1}^{-1}}\{s_{k_{0}}\}\not\in int(D)$, which contradicts with $(\ref{bbb})$. Hence 	 $\omega_{0} \not\in \{\omega: \omega \in \mathcal{A}_{k}, i.o.\}$, and thus $\mathcal{H}\subseteq \{\omega: \omega \in \mathcal{A}_{k}, i.o.\}^{c}$,  which means $P\{\omega: \omega \in\mathcal{A}_{k}, i.o.\}=0$.
\end{proof}	

We now give the proof of Theorem $\ref{thm3}$.

	\noindent\hspace{2em}{\itshape Proof of Theorem \ref{thm3}: }
	Let $p_{i}=\frac{\beta_{i}}{\sqrt{\sigma_{i}(\phi_{i}^{\top}\hat{\theta}_{i})}}$, $q_{i}=\frac{G'_{i}(\phi_{i}^{\top}\theta)}{\sqrt{\sigma_{i}(\phi_{i}^{\top}\theta)}}$. Then by $(\ref{P-1})$ we have $P_{n+1}=(\sum_{i=1}^{n}p_{i}^{2}\phi_{i}\phi_{i}^{\top}+P_{0}^{-1})^{-1}$, $Q_{n+1}=(\sum_{i=1}^{n}q_{i}^{2}\phi_{i}\phi_{i}^{\top}+P_{0}^{-1})^{-1}$. We first prove that $q_{i}/p_{i}\rightarrow 1, \;a.s.$ For this, we need only to show $\frac{\sigma_{k}(\phi_{k}^{\top}\bar{\theta}_{k})}{\sigma_{k}(\phi_{k}^{\top}\theta)}\rightarrow 1,\;a.s.$ and $\frac{\beta_{k}}{G'_{k}(\phi_{k}^{\top}\theta)}\rightarrow 1,\;a.s$. Notice that
		\begin{equation}\label{s}
		\begin{aligned}
			&\|\frac{\sigma_{k}(\phi_{k}^{\top}\hat{\theta}_{k})}{\sigma_{k}(\phi_{k}^{\top}\theta)}-1\|
			=\|\frac{\sigma_{k}(\phi_{k}^{\top}\hat{\theta}_{k})-\sigma_{k}(\phi_{k}^{\top}\theta)}{\sigma_{k}(\phi_{k}^{\top}\theta)}\|,\;\;\;a.s.
		\end{aligned}	
	\end{equation}	
Let $T_{k}=S_{k}(\phi_{k}^{\top}\hat{\theta}_{k}+e_{k+1})-S_{k}(\phi_{k}^{\top}\theta+e_{k+1})$ and  $B_{k}=\{\omega:\phi_{k}^{\top}\hat{\theta}_{k}\geq \phi_{k}^{\top}\theta \}$, we have $I_{B_{k}}$ is $\mathcal{F}_{k}-$measurable. By Assumption $\ref{assum2}$, it is not difficult to obtain that
\begin{equation}
\begin{aligned}
	0<|T_{k}|
	\leq2c+|\phi_{k}^{\top}\tilde{\theta}_{k}|\leq 2(c+M_{k})=O(1),\;\;\;a.s.
	\end{aligned}
\end{equation}
Moreover, by the fact that  $S_{k}(\cdot)$ is monotonically increasing and Assumption $\ref{assum4}$, we have
\begin{equation}
		\begin{aligned}
	&\mathbb{E}_{k}\left[|T_{k}|^{2}\right]\leq 2(c+M_{k})\mathbb{E}_{k}\left[|T_{k}|\right]\\
	\leq &2(c+M_{k})\left[\mathbb{E}_{k}\left[T_{k}I_{B_{k}}\right]+\mathbb{E}_{k}\left[-T_{k}I_{B_{k}^{c}}\right]\right]\\
	=&2(c+M_{k})[G_{k}(\phi_{k}^{\top}\hat{\theta}_{k})-G_{k}(\phi_{k}^{\top}\theta)]I_{B_{k}} \\
	&+	2(c+M_{k})[G_{k}(\phi_{k}^{\top}\theta)-G_{k}(\phi_{k}^{\top}\hat{\theta}_{k})]I_{B_{k}^{c}} \\
	=&2(c+M_{k})|G_{k}(\phi_{k}^{\top}\hat{\theta}_{k})-G_{k}(\phi_{k}^{\top}\theta)|(I_{B_{k}}+I_{B_{k}^{c}})\\
	\leq&2(c+M_{k})\overline{g}_{k}|\phi_{k}^{\top}\tilde{\theta}_{k}|=O(|\phi_{k}^{\top}\tilde{\theta}_{k}|),\;\;\;a.s.
		\end{aligned}
\end{equation} 
Thus by Theorem $\ref{thm2}$ and the condition $(\ref{phi})$,  we have 
\begin{equation}\label{oooo}
	\mathbb{E}_{k}\left[|T_{k}|^{2}\right]=o(1),\;\;\;a.s.
\end{equation}
Furthermore, by $(\ref{wwww})$ in Appendix \ref{A}, we have 
\begin{equation}\label{oo}
\mathbb{E}_{k}\left[[S_{k}(\phi_{k}^{\top}\hat{\theta}_{k}+e_{k+1})+S_{k}(\phi_{k}^{\top}\theta+e_{k+1})]^{2}\right]=O(1),\;a.s.	
\end{equation}
Hence, by Cauchy-Schwarz inequality, $(\ref{oooo})$ and $(\ref{oo})$, we have
\begin{equation}
		\begin{aligned}
	&\left|\mathbb{E}_{k}[S_{k}^{2}(\phi_{k}^{\top}\hat{\theta}_{k}+e_{k+1})-S_{k}^{2}(\phi_{k}^{\top}\theta+e_{k+1})]\right|^{2}\\
	\leq& \mathbb{E}_{k}\left[|S_{k}(\phi_{k}^{\top}\hat{\theta}_{k}+e_{k+1})-S_{k}(\phi_{k}^{\top}\theta+e_{k+1})|^{2}\right]\cdot \\
	&\mathbb{E}_{k}\left[|S_{k}(\phi_{k}^{\top}\hat{\theta}_{k}+e_{k+1})+S_{k}(\phi_{k}^{\top}\theta+e_{k+1})|^{2}\right]\\
	= & o(1),\;a.s.
	\end{aligned}
\end{equation}
Therefore, by the definition of $\sigma_{k}(\cdot)$, we have
\begin{equation}\label{dddd}
		\begin{aligned}
	&|\sigma_{k}(\phi_{k}^{\top}\hat{\theta}_{k})-\sigma_{k}(\phi_{k}^{\top}\theta)|\\
	\leq&|\mathbb{E}_{k}[S_{k}^{2}(\phi_{k}^{\top}\hat{\theta}_{k}+e_{k+1})-S_{k}^{2}(\phi_{k}^{\top}\theta+e_{k+1})]|\\
	&+|G_{k}^{2}(\phi_{k}^{\top}\hat{\theta}_{k})-G_{k}^{2}(\phi_{k}^{\top}\theta)|=o(1),\;\;a.s.
		\end{aligned}
\end{equation}
Hence by $(\ref{muu})$, $(\ref{s})$ and  $(\ref{dddd})$, we obtain 
\begin{equation}\label{fin}
\frac{\sigma_{k}(\phi_{k}^{\top}\hat{\theta}_{k})}{\sigma_{k}(\phi_{k}^{\top}\theta)}\rightarrow 1,\;\;a.s.
\end{equation}
	Besides, from Lagrange mean value theorem and the definition of $\beta_{k}$, we can easily have 
	\begin{equation}
		|\beta_{k}-G'_{k}(\phi_{k}^{\top}\theta)|\leq \rho\max(|\phi_{k}^{\top}\tilde{\theta}_{k}|, |\phi_{k}^{\top}\tilde{\bar{\theta}}_{k}|))
	\end{equation}
	where $\rho$ is the Lipschitz constant of $G'_{k}(\cdot)$ as defined in Assumption $\ref{assum4}$. Thus by $(\ref{phi})$, $(\ref{22})$ and $(\ref{bb})$, we have
	\begin{equation}\label{bet}
		\begin{aligned}
			&|\frac{\beta_{k}}{G'_{k}(\phi_{k}^{\top}\theta)}-1|=|\frac{\beta_{k}-G'_{k}(\phi_{k}^{\top}\theta)}{G'_{k}(\phi_{k}^{\top}\theta)}|\\
			\leq&\frac{\rho\max(|\phi_{k}^{\top}\tilde{\theta}_{k}|, |\phi_{k}^{\top}\tilde{\bar{\theta}}_{k}|))}{\underline{g}_{k}}
			=O(\sqrt{\frac{\log k}{\lambda_{\min}(k)}})=o(1).\;\;\;a.s.
		\end{aligned}
	\end{equation}
	From $(\ref{fin})$ and $(\ref{bet})$, we finally have $q_{i}/p_{i} \rightarrow 1$ almost surely. Hence by $(\ref{phi})$ and Lemma \ref{lemp}, we obtain that
	
	\begin{equation}\label{PQ}
		Q_{k+1}^{-\frac{1}{2}}P_{k+1}Q_{k+1}^{-\frac{1}{2}} \rightarrow I,\;\;a.s.
	\end{equation}
	Moreover, from $(\ref{sss})$ and $(\ref{aalpha})$, we have
	\begin{equation}
			\begin{aligned}
		\theta-s_{k}=&	\theta-\hat{\theta}_{k}-a_{k}\beta_{k}P_{k}\phi_{k}[y_{k+1}-G_{k}(\phi_{k}^{\top}\hat{\theta}_{k})]\\
		=&\tilde{\theta}_{k}-a_{k}\beta_{k}P_{k}\phi_{k}(\psi_{k}+w_{k+1})\\
		=&(I-\mu_{k}^{-1}\beta_{k}\xi_{k}P_{k+1}\phi_{k}\phi_{k}^{\top})\tilde{\theta}_{k}-\mu_{k}^{-1}\beta_{k}P_{k+1}\phi_{k}w_{k+1},
				\end{aligned}
	\end{equation}	
where  $w_{k+1}$ is defined in $(\ref{star})$, and
\begin{equation} \xi_{k}=\frac{\psi_{k}}{\phi_{k}^{\top}\tilde{\theta}_{k}}I_{\{\phi_{k}^{\top}\tilde{\theta}_{k}\not=0\}}+G'(\phi_{k}^{\top}\hat{\theta}_{k})I_{\{\phi_{k}^{\top}\tilde{\theta}_{k}=0\}}.
\end{equation}	
Furthermore, from $(\ref{be2})$, we have
	\begin{equation}
		\begin{aligned}
			\hat{\theta}_{k+1}=&s_{k}I_{\{s_{k}\in D\}}+\Pi_{P_{k+1}^{-1}}\{s_{k}\}\cdot I_{\{s_{k}\not\in D\}}\\
			=&s_{k}-(s_{k}-\Pi_{P_{k+1}^{-1}}\{s_{k}\})I_{\{s_{k}\not\in D\}}.			
		\end{aligned}
	\end{equation}

	Thus, we obtain that
		\begin{equation}\label{Pr}
		\begin{aligned}
			\tilde{\theta}_{k+1}	=&\theta-s_{k}+(s_{k}-\Pi_{P_{k+1}^{-1}}\{s_{k}\})I_{\{s_{k}\not\in D\}}\\		
			=&(I-\mu_{k}^{-1}\beta_{k}\xi_{k}P_{k+1}\phi_{k}\phi_{k}^{\top})\tilde{\theta}_{k}-\mu_{k}^{-1}\beta_{k}P_{k+1}\phi_{k}w_{k+1}\\
			&+(s_{k}-\Pi_{P_{k+1}^{-1}}\{s_{k}\})I_{\{s_{k}\not\in D\}}\\	
			=&(I-\mu_{k}^{-1}\beta_{k}^{2}P_{k+1}\phi_{k}\phi_{k}^{\top})\tilde{\theta}_{k}-\mu_{k}^{-1}\beta_{k}P_{k+1}\phi_{k}w_{k+1}\\
			&-\mu_{k}^{-1}\beta_{k}(\xi_{k}-\beta_{k})P_{k+1}\phi_{k}\phi_{k}^{\top}\tilde{\theta}_{k}\\
				&+(s_{k}-\Pi_{P_{k+1}^{-1}}\{s_{k}\})I_{\{s_{k}\not\in D\}}\\
			=&P_{k+1}P_{k}^{-1}\tilde{\theta}_{k}-P_{k+1}\mu_{k}^{-1}\beta_{k}(\xi_{k}-\beta_{k})\phi_{k}\phi_{k}^{\top}\tilde{\theta}_{k}\\
			&-P_{k+1}\mu_{k}^{-1}\beta_{k}\phi_{k}w_{k+1}
			+(s_{k}-\Pi_{P_{k+1}^{-1}}\{s_{k}\})I_{\{s_{k}\not\in D\}}.
		\end{aligned}
	\end{equation}
From $(\ref{Pr})$, we have
	\begin{equation}\label{kkk}
			\begin{aligned}
		&P_{k+1}^{-1}\tilde{\theta}_{k+1}\\
		=&P_{k}^{-1}\tilde{\theta}_{k}-\mu_{k}^{-1}\beta_{k}\phi_{k}w_{k+1}-\mu_{k}^{-1}\beta_{k}(\xi_{k}-\beta_{k})\phi_{k}\phi_{k}^{\top}\tilde{\theta}_{k}\\
		&+P_{k+1}^{-1}(s_{k}-\Pi_{P_{k+1}^{-1}}\{s_{k}\})I_{\{s_{k}\not\in D\}}\\
		=&P_{0}^{-1}\tilde{\theta}_{0}-\sum_{i=0}^{k}\mu_{i}^{-1}\beta_{i}\phi_{i}w_{i+1}-\sum_{i=0}^{k}\mu_{i}^{-1}\beta_{i}(\xi_{i}-\beta_{i})\phi_{i}\phi_{i}^{\top}\tilde{\theta}_{i}\\
		&+\sum_{i=0}^{k}P_{i+1}^{-1}(s_{i}-\Pi_{P_{i+1}^{-1}}\{s_{i}\})I_{\{s_{i}\not\in D\}}.
			\end{aligned}
		\end{equation}	
	Therefore,  multiplying $Q_{k+1}^{-\frac{1}{2}}P_{k+1}$ from the left, we have
	\begin{equation}\label{P}
		\begin{aligned}
			Q_{k+1}^{-\frac{1}{2}}\tilde{\theta}_{k+1}
			=&Q_{k+1}^{-\frac{1}{2}}P_{k+1}P_{0}^{-1}\tilde{\theta}_{0}-Q_{k+1}^{-\frac{1}{2}}P_{k+1 }\sum_{i=0}^{k}\mu_{i}^{-1}\beta_{i}\phi_{i}w_{i+1}\\
			&-Q_{k+1}^{-\frac{1}{2}}P_{k+1}\sum_{i=0}^{k}\mu_{i}^{-1}\beta_{i}(\xi_{i}-\beta_{i})\phi_{i}\phi_{i}^{\top}\tilde{\theta}_{i}\\
			&+Q_{k+1}^{-\frac{1}{2}}P_{k+1}\sum_{i=0}^{k}P_{i+1}^{-1}(s_{i}-\Pi_{P_{i+1}^{-1}}\{s_{i}\})I_{\{s_{i}\not\in D\}}.
		\end{aligned}
	\end{equation}
	We now proceed to show that the main term on the RHS of $(\ref{P})$ is the second term and other terms can be neglected asymptotically.  First, by $(\ref{PQ})$ and the fact that $\|P_{k+1}^{\frac{1}{2}}\|\rightarrow 0$  almost surely by $(\ref{phi})$, we have
	\begin{equation}\label{k1}
		\|Q_{k+1}^{-\frac{1}{2}}P_{k+1}P_{0}^{-1}\tilde{\theta}_{0}\|\leq \|Q_{k+1}^{-\frac{1}{2}}P_{k+1}^{\frac{1}{2}}\| \|P_{k+1}^{\frac{1}{2}}\|\|P_{0}^{-1}\tilde{\theta}_{0}\|\rightarrow 0,\;a.s.
	\end{equation}
Next, by Lemma $\ref{lempp}$, for any $\omega \in \mathcal{H}$, we have $I_{\{s_{i}\not\in D\}}\not=0$ for only finite number of $i$. Therefore, $\|\sum_{i=0}^{k}P_{i+1}^{-1}(s_{i}-\Pi_{P_{i+1}^{-1}}\{s_{i}\})I_{\{s_{i}\not\in D\}}\|$ is finite almost surely as $k$ tends to infinity. Hence
\begin{equation}\label{k2}
			\begin{aligned}
	&\|Q_{k+1}^{-\frac{1}{2}}P_{k+1}\sum_{i=0}^{k}P_{i+1}^{-1}(s_{i}-\Pi_{P_{i+1}^{-1}}\{s_{i}\})I_{\{s_{i}\not\in D\}}\|\\
	\leq & \|Q_{k+1}^{-\frac{1}{2}}P_{k+1}^{\frac{1}{2}}\| \|P_{k+1}^{\frac{1}{2}}\| \|\sum_{i=0}^{k}P_{i+1}^{-1}(s_{i}-\Pi_{P_{i+1}^{-1}}\{s_{i}\})I_{\{s_{i}\not\in D\}}\|\\
	&\rightarrow 0,\;\;a.s.,
			\end{aligned}
\end{equation}
	 Moreover, from Lagrange mean value theorem, for any $k \geq 0$, there exist $\iota_{k} \in R$ and $\kappa_{k} \in R$, where $\iota_{k}$ is between $\phi_{k}^{\top}\theta$ and $\phi_{k}^{\top}\hat{\theta}_{k}$, $\kappa_{k}$ is between $\phi_{k}^{\top}\bar{\theta}_{k}$ and $\phi_{k}^{\top}\hat{\theta}_{k}$,  such that $\xi_{k}=G_{k}'(\iota_{k})$ and $\beta_{k}=G_{k}'(\kappa_{k})$. Therefore, we obtain that
	 \begin{equation}\label{73}
	 	\begin{aligned}
	 		|\xi_{k}-\beta_{k}|&=|G_{k}'(\iota_{k})-G_{k}'(\kappa_{k})|\leq \rho|\iota_{k}-\kappa_{k}|\\
	 		&\leq \rho|\iota_{k}-\phi_{k}^{\top}\theta|+\rho|\phi_{k}^{\top}\theta-\kappa_{k}|\\	
	 		&\leq 2\rho\cdot\max(|\phi_{k}^{\top}\tilde{\bar{\theta}}_{k}|,|\phi_{k}^{\top}\tilde{\theta}_{k}|),
	 	\end{aligned}
	 \end{equation}
	 where $\rho$ is the lipschitz constant as defined in Assumption $\ref{assum4}$.	Thus, we have
	 	\begin{equation}\label{lei}
		\begin{aligned}
			&\|Q_{k+1}^{-\frac{1}{2}}P_{k+1}\sum_{i=0}^{k}\mu_{i}^{-1}\beta_{i}(\xi_{i}-\beta_{i})\phi_{i}\phi_{i}^{\top}\tilde{\theta}_{i}\|\\
			=&O( \|Q_{k+1}^{-\frac{1}{2}}P_{k+1}^{\frac{1}{2}}\|\cdot \|P_{k+1}^{\frac{1}{2}}\| \sum_{i=0}^{k}\|\phi_{i}^{\top}\tilde{\bar{\theta}}_{i}\|^{2})\\
			&+O( \|Q_{k+1}^{-\frac{1}{2}}P_{k+1}^{\frac{1}{2}}\|\cdot \|P_{k+1}^{\frac{1}{2}}\| \sum_{i=0}^{k}\|\phi_{i}^{\top}\tilde{\theta}_{i}\|^{2}),
		\end{aligned}	
	\end{equation}
	where we have used the fact that $\mu_{i}^{-1}\beta_{i}\|\phi_{i}\|$ are bounded from above. By $(\ref{ffr})$, $(\ref{22})$ and the condition $(\ref{phi})$,  we conclude that the RHS of the $(\ref{lei})$ tends to $0$ almost surely. Similar to the reasons explained in the proof of Lemma 1, the analysis of the third term on the RHS of $(\ref{P})$ is an essential feature of the current two-step identification algorithm.
	
	Now, it only remains to consider the second term of the RHS of $(\ref{P})$. By $(\ref{PQ})$ and $(\ref{R1})$, we have 
	\begin{equation}\label{noi}
	Q_{k+1}^{-\frac{1}{2}}P_{k+1}\Delta_{k+1}=(Q_{k+1}^{-\frac{1}{2}}P_{k+1}Q_{k+1}^{-\frac{1}{2}})(Q_{k+1}^{\frac{1}{2}}\Delta_{k+1})\mathop{\rightarrow}\limits^{p}I.
	\end{equation}
	Hence by  $(\ref{noi})$,  we need only to show that
	\begin{equation}\label{R}
		\Delta_{k+1}^{-1}\sum_{i=0}^{k}\mu_{i}^{-1}\beta_{i}\phi_{i}w_{i+1}\mathop{\rightarrow}\limits^{d}N(0,I).
	\end{equation}
	We now prove that for any non-random $m\times 1$ vector $v$, 
	\begin{equation}\label{c}
		v^{\top}\Delta_{k+1}^{-1}\sum_{i=0}^{k}\mu_{i}^{-1}\beta_{i}\phi_{i}w_{i+1}\mathop{\rightarrow}\limits^{d}N(0,\|v\|^{2}),
	\end{equation}
	thus by Cramer-Wold device, $(\ref{R})$ holds. 
	
	To prove $(\ref{c})$, let $x_{k,i}=v^{\top}\Delta_{k+1}^{-1}\mu_{i}^{-1}\beta_{i}\phi_{i}$ and $x'_{k,i}=x_{k,i}I\{\|x_{k,i}\|\leq 1\}$. Thus, $x'_{k,i}$ is $\mathcal{F}_{i}-$measurable and $\mathbb{E}_{i}\{x'_{k,i}w_{i+1}\}=0$. Besides, by Assumption $\ref{assum4}$ and $(\ref{wwww})$, we can easily have $x'_{k,i}w_{i+1}$ is an $\mathcal{L}_{2}$ sequence . Moreover, by $(\ref{phi})$ and $(\ref{R1})$, it is not difficult to obtain
$
\max\limits_{1\leq i \leq k}\{|x_{k,i}|\}\underset{k \rightarrow \infty}{\stackrel{p}{\longrightarrow}}0.
$	
Hence, we have $P\{\max\limits_{1\leq i \leq k}|x_{k,i}|> 1\}\underset{k \rightarrow \infty}{\longrightarrow}0,$ which means
\begin{equation}\label{pppp}
P\{x_{k,i}\not=x'_{k,i}, \text{for some} 1\leq i \leq k\}\underset{k \rightarrow \infty}{\longrightarrow}0.	
\end{equation}
Furthermore,  by definitions of $w_{i+1}$ in  $(\ref{star})$ and $\sigma_{i}(\cdot)$ in $(\ref{sig})$, and the fact that $\phi_{k}^{\top}\theta$ is $\mathcal{F}_{k}-$measurable,  we  have $\mathbb{E}[w_{i+1}^{2}\mid \mathcal{F}_{i}]=\sigma_{i}(\phi_{i}^{\top}\theta)$. Therefore, 
\begin{equation}\label{E}
	\begin{aligned}
		&\sum_{i=0}^{k}\mathbb{E}[(v^{\top}\Delta_{k+1}^{-1}\mu_{i}^{-1}\beta_{i}\phi_{i}w_{i+1})^{2}\mid \mathcal{F}_{i}]\\
		=&\sum_{i=0}^{k}v^{\top}\Delta_{k+1}^{-1}\mu_{i}^{-2}\beta_{i}^{2}\sigma_{i}(\phi_{i}^{\top}\theta)\phi_{i}\phi_{i}^{\top}\Delta_{k+1}^{-1}v\\
		=&v^{\top}\Delta_{k+1}^{-1}Q_{k}^{-\frac{1}{2}}Q_{k}^{\frac{1}{2}}[\sum_{i=0}^{k}b_{i}^{2}\frac{(G'(\phi_{i}^{\top}\theta))^{2}}{\sigma_{i}(\phi_{i}^{\top}\theta)}\phi_{i}\phi_{i}^{\top}]Q_{k}^{\frac{1}{2}}Q_{k}^{-\frac{1}{2}}\Delta_{k+1}^{-1}v\\
		&\underset{k \rightarrow \infty}{\stackrel{p}{\longrightarrow}}\|v\|^{2}. 
	\end{aligned}
\end{equation}
where the last step is from Lemma $\ref{lemp}$ and $(\ref{R1})$, and where $b_{i}=\frac{\beta_{i}\sigma_{i}(\phi_{i}^{\top}\theta)}{\sigma_{i}(\phi_{i}^{\top}\hat{\theta}_{i})G'(\phi_{i}^{\top}\theta)}$, which converges to $1$ almost surely by $(\ref{s})$ and $(\ref{bet})$. By $(\ref{pppp})$ and $(\ref{E})$, we then have 
\begin{equation}\label{fff}
	\sum_{i=0}^{k}\mathbb{E}[(x'_{k,i}w_{i+1})^{2}\mid \mathcal{F}_{i}]\underset{k \rightarrow \infty}{\stackrel{p}{\longrightarrow}}\|v\|^{2}. 
\end{equation}
Using $(\ref{fff})$ and Assumption $\ref{assum4}$, for any $\epsilon > 0$, it can be shown that
\begin{equation}
	\begin{aligned}
		&\sum_{i=1}^{k}\mathbb{E}_{i}\{(x'_{k,i}w_{i+1})^{2}I_{\{|x'_{k,i}w_{i+1}|>\epsilon\}}\}\\
		\leq&\sum_{i=1}^{k}\mathbb{E}_{i}\{(x'_{k,i}w_{i+1})^{2}\frac{|x'_{k,i}w_{i+1}|^{\eta}}{\epsilon^{\eta}}\}\\
		= &\sum_{i=1}^{k}\mathbb{E}_{i}\{|x'_{k,i}|^{2+\eta}\frac{|w_{i+1}|^{2+\eta}}{\epsilon^{\eta}}\}\\
		= &\sum_{i=1}^{k}\mathbb{E}_{i}\{(x'_{k,i}w_{i+1})^{2}\}
		\frac{|x'_{k,i}|^{\eta}\mathbb{E}_{i}\{|w_{i+1}|^{2+\eta}\}}{\epsilon^{\eta}\mathbb{E}_{i}\{|w_{i+1}|^{2}\}} \\
		= & O(\max\limits_{1\leq i \leq k}|x'_{k,i}|^{\eta}\sum_{i=0}^{k}\mathbb{E}_{i}\{(x'_{k,i}w_{i+1})^{2}\})\underset{k \rightarrow \infty}{\stackrel{p}{\longrightarrow}}0,
	\end{aligned}
\end{equation}
Hence by Lemma $\ref{lem9}$, we have
\begin{equation}\label{1022}			\sum_{i=0}^{k}x'_{k,i}w_{i+1}\mathop{\rightarrow}\limits^{d}N(0,\|v\|^{2}).
\end{equation}
Thus by $(\ref{pppp})$ and $(\ref{1022})$, $(\ref{c})$ holds. To conclude, combining $(\ref{P})$, $(\ref{k1})$, $(\ref{k2})$, $(\ref{lei})$ and $(\ref{R})$, we finally obtain the result $(\ref{theta})$.
	\hspace*{\fill}~\QED\par\endtrivlist\unskip

\noindent\hspace{2em}{\itshape Proof of Theorem $\ref{thm4}$: }	
Let 	
	\begin{equation}\label{102}                                 
		\begin{aligned}
			Y_{k+1}=&-2a_{k}\psi_{k}w_{k+1}+2\mu_{k}^{-1}(\bar{\psi}_{k}-\beta_{k}\phi_{k}^{\top}\tilde{\bar{\theta}}_{k})w_{k+1}\\
			&+\mu_{k}^{-1}a_{k}\beta_{k}^{2}\phi_{k}^{\top}P_{k}\phi_{k}\left(w_{k+1}^{2}-\mathbb{E}_{k}[w_{k+1}^{2}]\right).
		\end{aligned}
	\end{equation}
If $\{\sum\limits_{k=0}^{n}Y_{k},\mathcal{F}_{n}\}$ is a $\mathcal{L}_{2}-$martingale, by Lemma \ref{lem8} with $x=\frac{1}{18\sigma_{b}}$ and $y=18 \sigma_{b}\frac{1-\alpha}{\alpha}$, where $\sigma_{b}=\sup\limits_{0\leq k\leq N}\mu_{k}^{-1}\mathbb{E}_{k}[|w_{k+1}|^{2}],$ we can easily obtain that with at least probability $1-\alpha$,
\begin{equation}\label{sigmab}
		\begin{aligned}	
			&\sum_{k=0}^{n}Y_{k+1} \leq \frac{\sum_{k=0}^{n}\mathbb{E}_{k}[|Y_{k+1}|^{2}]}{18\sigma_{b}}+18\sigma_{b}\frac{1-\alpha}{\alpha},\;\forall 0\leq n\leq N.
		\end{aligned}
	\end{equation}
	If $\{\sum\limits_{k=0}^{n}Y_{k},\mathcal{F}_{n}\}$ is not a $\mathcal{L}_{2}-$martingale, let $x_{k}=\mu_{k}^{-1}(\bar{\psi}_{k}-\beta_{k}\phi_{k}^{\top}\tilde{\bar{\theta}}_{k})-a_{k}\psi_{k}$, $x'_{k}=\mu_{k}^{-1}a_{k}\beta_{k}^{2}\phi_{k}^{\top}P_{k}\phi_{k}.$ For $j=1,2,\cdots,$ let $x_{j,k}=x_{k}I_{\{|x_{k}|\leq j\}}$, $x'_{j,k}=x'_{k}I_{\{|x'_{k}|\leq j\}}$ and $Y'_{j,k+1}=x_{j,k}w_{k+1}+x'_{j,k}\left(w_{k+1}^{2}-\mathbb{E}_{k}[w_{k+1}^{2}]\right).$ Thus, $\{\sum_{k=0}^{n}Y'_{j,k},\mathcal{F}_{n}\}$ is an $\mathcal{L}_{2}-$martingale for $j=1,2,\cdots$. Hence, $(\ref{sigmab})$ can also be obtained from the fact that $\{x_{k}\}$ and $\{x'_{k}\}$ are bounded sequence almost surely and
$\lim\limits_{j\rightarrow \infty} P\{x_{j,k}\not=x_{k}\; \text{or}\; x'_{j,k}\not=x'_{k}, \text{for some}\; k\geq 0\}=0.$

Notice that
	\begin{equation}\label{mar}
		\begin{aligned}
			&\mathbb{E}[|Y_{k+1}|^{2}\mid \mathcal{F}_{k}]\leq 9a_{k}\mu_{k}^{-1}\psi_{k}^{2}\mathbb{E}_{k}[|w_{k+1}|^{2}]\\
			&+9\mu_{k}^{-2}(\bar{\psi}_{k}-\beta_{k}\phi_{k}^{\top}\tilde{\bar{\theta}}_{k})^{2}\mathbb{E}_{k}[|w_{k+1}|^{2}]\\
			&+9\mu_{k}^{-2}(a_{k}\beta_{k}^{2}\phi_{k}^{\top}P_{k}\phi_{k})^{2}\mathbb{E}_{k}\left[\left(w_{k+1}^{2}-\mathbb{E}_{k}[|w_{k+1}|^{2}]\right)^{2}\right],\;\;a.s.
		\end{aligned}
	\end{equation}
	Besides, following the similar proof idea of Remark 3.2 in \cite{g1995}, we have 
	\begin{equation}\label{a2}
		\begin{aligned}
			\sum_{k=0}^{n}(a_{k}\beta_{k}^{2}\phi_{k}^{\top}P_{k}\phi_{k})^{2}=&\sum_{k=0}^{n}a_{k}(\beta_{k}\phi_{k}^{\top})(P_{k}-P_{k+1})(\beta_{k}\phi_{k})\\
			\leq& \Phi\sum_{k=0}^{n}tr(P_{k}-P_{k+1})\leq \Phi tr(P_{0}),
		\end{aligned}
	\end{equation}
	where $\Phi=\sup\limits_{0\leq k\leq N}\{\mu_{k}^{-1}\beta_{k}^{2}\|\phi_{k}\|^{2}\}$.   Thus, by $(\ref{VV})$, $(\ref{sigmab})$, $(\ref{mar})$, $(\ref{a2})$ and Lemma $\ref{lem3}$, we know that the following holds with probability at least $1-\alpha$:
	\begin{equation}\label{52}
		\begin{aligned}
			V_{n+1}+\frac{1}{2}\sum_{k=0}^{n}a_{k}\psi_{k}^{2}
			 \leq &\sigma_{b}\log |P_{n+1}^{-1}|+\Gamma+\frac{3}{2}S_{n},\;\forall 0\leq n\leq N.
		\end{aligned}
	\end{equation}
	where $\Gamma=V_{0}+\sigma_{b}\log|P_{0}^{-1}|+\frac{\Phi tr(P_{0})\bar{\sigma}_{b}}{2\sigma_{b}}+18 \sigma_{b}\frac{1-\alpha}{\alpha}$, $\overline{\sigma}_{b}=\sup\limits_{0\leq k\leq N}\mu_{k}^{-2}\mathbb{E}_{k}\left[(w_{k+1}^{2}-\mathbb{E}_{k}[w_{k+1}^{2}])^{2}\right],$ and 
	\begin{equation}\label{ss}
			S_{n}=\sum_{k=0}^{n} \mu_{k}^{-1}(\bar{\psi}_{k}-\beta_{k}\phi_{k}^{\top}\tilde{\bar{\theta}}_{k})^{2},	
	\end{equation}	
	
	To analyze the term $S_{n}$ above, we consider the Lyapunov function
	$\bar{V}_{k}=\tilde{\bar{\theta}}_{k}^{\top}\bar{P}_{k}^{-1}\tilde{\bar{\theta}}_{k}.$ Using the similar analysis above and the  Lyapunov function analysis as in $(\ref{558})$-$(\ref{com})$, we can obtain that with probability at least $1-\alpha$, for $0\leq n \leq N$,	
	\begin{equation}\label{sigmaa}
		\begin{aligned}
			\bar{V}_{n+1}+\frac{1}{4}\sum_{k=0}^{n}\bar{a}_{k}\bar{\beta}_{k}^{2}(\phi_{k}^{\top}\tilde{\bar{\theta}}_{k})^{2}
			 \leq \bar{\mu}\sigma_{b}\log |\bar{P}_{n+1}^{-1}|+\bar{\Gamma},
		\end{aligned}
	\end{equation}
	where $\bar{\Gamma}=\bar{V}_{0}+\bar{\mu}\sigma_{b}\log |\bar{P}_{0}^{-1}|+\frac{\bar{\Phi}tr(\bar{P}_{0})\bar{\mu}\bar{\sigma}_{b}}{2\sigma_{b}}+10\bar{\mu}\sigma_{b}\frac{1-\alpha}{\alpha}$.		
		For simplicity, denote
 $T_{n}=V_{n+1}+\frac{1}{2}\sum_{k=0}^{n}a_{k}\psi_{k}^{2}$, and $ \bar{T}_{n}=\bar{V}_{n+1}+\frac{1}{4}\sum_{k=0}^{n}\bar{a}_{k}\bar{\beta}_{k}^{2}(\phi_{k}^{\top}\tilde{\bar{\theta}}_{k})^{2}$. Besides,  let
	\begin{equation}\label{e33}
		\begin{aligned}
			&E_{1}=\{\omega: \bar{T}_{n} > \bar{\mu}\sigma_{b}\log |\bar{P}_{n+1}^{-1}|+\bar{\Gamma}, \forall\; 0 \leq n\leq N\},\\
			&E_{2}
			=\{\omega: T_{n} > \sigma_{b}\log |P_{n+1}^{-1}|+\Gamma+\frac{3}{2}S_{n},\forall\; 0 \leq n\leq N\}.\\
		\end{aligned}
	\end{equation}
	From $(\ref{ss})$ and $(\ref{sigmaa})$, for any $\omega\in E_{1}^{c}$, we have
	\begin{equation}
			\begin{aligned}
		\frac{3}{2}S_{n}\leq&\sum_{k=0}^{n}\frac{3}{2} \mu_{k}^{-1}\overline{g}_{k}^{2}(\phi_{k}^{\top}\tilde{\bar{\theta}}_{k})^{2}\leq 6\gamma \bar{T}_{n}\\
		\leq & 6\gamma (\bar{\mu}\sigma_{b}\log |P_{n+1}^{-1}|+\bar{\Gamma}).
			\end{aligned}
	\end{equation}
where $\gamma=\sup\limits_{0\leq k \leq N}\frac{\mu_{k}^{-1}\overline{g}_{k}^{2}}{\bar{a}_{k}^{2}\bar{\beta}_{k}^{2}}$.
	Therefore, for any $\omega\in E_{1}^{c}\cap E_{2}^{c}$, we obtain for any $0\leq n \leq N$
	\begin{equation}\label{ccc}
		\begin{aligned}
		T_{n}\leq&  \sigma_{b}\log|P_{n+1}^{-1}|+6\gamma\bar{\mu}\sigma_{b}\log|\bar{P}_{n+1}^{-1}|+\Gamma+ 6\gamma\bar{\Gamma}\\
		\leq&(\sigma_{b}+6\gamma\bar{\mu}^{2}\sigma_{b})\log|P_{n+1}^{-1}|+\Gamma+6\gamma\bar{\Gamma},
	\end{aligned}
	\end{equation} 
where we have used the fact that $\log|\bar{P}_{n+1}^{-1}|<\bar{\mu}\log|P_{n+1}^{-1}|$.

Moreover, by $(\ref{barpsi})$, $(\ref{73})$ and $(\ref{ss})$, we also have
\begin{equation}\label{152}
		\begin{aligned}
	\frac{3}{2}S_{N}\leq &c_{0}+\sum_{k=1}^{N}\frac{6\rho^{2} }{\mu_{k}}\max(|\phi_{k}^{\top}\tilde{\bar{\theta}}_{k}|^{4}, |\phi_{k}^{\top}\tilde{\theta}_{k}|^{4})\\
	\leq &c_{0}+\Psi \sum_{k=1}^{N}(\|\tilde{\bar{\theta}}_{k} \|^{2}\bar{a}_{k}\bar{\beta}_{k}^{2}|\phi_{k}^{\top}\tilde{\bar{\theta}}_{k}|^{2}+\|\tilde{\theta}_{k}\|^{2}a_{k}\bar{\beta}_{k}^{2}|\phi_{k}^{\top}\tilde{\theta}_{k}|^{2}),
\end{aligned}
\end{equation}	
where $c_{0}=\frac{3}{2}\mu_{0}^{-1}\overline{g}_{0}^{2}(\phi_{0}^{\top}\tilde{\bar{\theta}}_{0})^{2}$ and $\Psi=\sup\limits_{0\leq k\leq N}\frac{6\mu_{k}^{-1}\rho^{2}\|\phi_{k}\|^{2}}{\bar{\beta}_{k}^{2} max\{a_{k},\bar{a}_{k}\}}$. We now analyze the RHS of $(\ref{152})$. Indeed, for any $\omega \in E_{1}^{c}$, we have for any $1\leq k \leq N$ and given $\tau>0$,
\begin{equation}
		\begin{aligned}
	\|\tilde{\bar{\theta}}_{k}\|^{2} \leq \frac{\bar{\mu}\sigma_{b}\log |\bar{P}_{k}^{-1}|+\bar{\Gamma}}{\lambda_{\min}\{\bar{P}_{k}^{-1}\}}
	\leq \frac{\bar{\mu}\sigma_{b}\log |\bar{P}_{k+1}^{-1}|+\bar{\Gamma}+1}{\lambda_{N}(\log|\bar{P}_{k+1}^{-1}|)^{2+\tau}}.
\end{aligned}
\end{equation}
where $\lambda_{N}=\inf\limits_{0\leq k \leq N}\{\frac{\lambda_{\min}\{\bar{P}_{k}^{-1}\}}{(\log|\bar{P}_{k+1}^{-1}|)^{2+\tau}},\frac{\lambda_{\min}\{P_{k}^{-1}\}}{(\log|P_{k+1}^{-1}|)^{2+\tau}}\}$. Let $D_{n}=\sum_{k=0}^{n}\bar{a}_{k}\bar{\beta}_{k}^{2}(\phi_{k}^{\top}\tilde{\bar{\theta}}_{k})^{2}, n=0, 1, \cdots$, we have $\frac{1}{4}D_{n}+1\leq \bar{\mu}\sigma_{b}\log |\bar{P}_{n+1}^{-1}|+\bar{\Gamma}+1 $ for any $\omega \in E_{1}^{c}$ and $0\leq n \leq N$. Thus, we obtain
\begin{equation}\label{1566}
		\begin{aligned}
	&\Psi\sum_{k=1}^{N}\|\tilde{\bar{\theta}}_{k} \|^{2}\bar{a}_{k}\bar{\beta}_{k}^{2}|\phi_{k}^{\top}\tilde{\bar{\theta}}_{k}|^{2}\\
	\leq&4\Psi\sum_{k=1}^{N}(\frac{1}{4}D_{k}-\frac{1}{4}D_{k-1})\frac{\bar{\mu}\sigma_{b}\log |\bar{P}_{k+1}^{-1}|+\bar{\Gamma}+1}{\lambda_{N}(\log|\bar{P}_{k+1}^{-1}|)^{2+\tau}}\\
	\leq& \frac{4\Psi}{\lambda_{N}}(\bar{\mu}\sigma_{b}+\bar{\Gamma}+1)^{2+\tau}\sum_{k=1}^{N}\int_{\frac{1}{4}D_{k-1}+1}^{\frac{1}{4}D_{k}+1}\frac{1}{t^{1+\tau}}dt\\
	\leq&\frac{4\Psi}{\tau\lambda_{N}}(\bar{\mu}\sigma_{b}+\bar{\Gamma}+1)^{2+\tau}.
\end{aligned}
\end{equation}	
	Similarly, by $(\ref{ccc})$,  we have for any $\omega \in E_{1}^{c}\cap E_{2}^{c}$, 
	\begin{equation}\label{1577}
	\begin{aligned}
		&\Psi\sum_{k=0}^{N}\|\tilde{\theta}_{k}\|^{2}a_{k}\beta_{k}^{2}|\phi_{k}^{\top}\tilde{\theta}_{k}|^{2}\\
		\leq &\frac{2\Psi}{\tau\lambda_{N}}[(1+6\gamma\bar{\mu}^{2})(\sigma_{b}+\Gamma+\frac{\bar{\Gamma}}{\bar{\mu}^{2}}+1)]^{2+\tau}.
	\end{aligned}
	\end{equation}
Hence, by $(\ref{52})$, $(\ref{152})$,  $(\ref{1566})$ and $(\ref{1577})$, for any $\omega \in E_{1}^{c}\cap E_{2}^{c},$ and $0\leq n\leq N$, we have 
\begin{equation}\label{155}
	V_{n+1}+\frac{1}{2}\sum_{k=0}^{N}a_{k}\psi_{k}^{2}\leq \sigma_{b}\log|P_{n+1}^{-1}|+\frac{2\Psi C}{\tau\lambda_{N}}+\Gamma+c_{0}.
\end{equation}	
 From  $(\ref{52})$ and $(\ref{sigmaa})$, we have $P(E_{1})\leq \alpha$ and $P(E_{2})\leq \alpha$. Thus,
\begin{equation}\label{1600}
	P(E_{1}^{c}\cap E_{2}^{c})\geq 1-2\alpha.
\end{equation}
	Therefore, $(\ref{regret})$ can be obtained from $(\ref{155})$ and $(\ref{1600})$. Furthermore, we have
	\begin{equation}\label{157}
		\begin{aligned}
			\|\tilde{\theta}_{k+1}^{ (j)}\|^{2}=&\|e_{j}^{\top}\bar{P}_{k+1}^{\frac{1}{2}}\cdot P_{k+1}^{-\frac{1}{2}}\tilde{\theta}_{k+1}\|^{2}
			\leq P_{k+1}^{(j)}\cdot V_{k+1},\;\;a.s.
		\end{aligned}
	\end{equation}
	hence (\ref{eeee}) is obtained from $(\ref{155})$, $(\ref{1600})$ and $(\ref{157})$. 	
	\hspace*{\fill}~\QED\par\endtrivlist\unskip
		
	\noindent\hspace{2em}{\itshape Proof of Proposition \ref{po1}: }
	For any given positive $\alpha$ and $t$ with $\alpha+t<1$, denote $\upsilon=\sqrt{\frac{\ln 2-\ln t}{2K}}$ and let
	\begin{equation}
		\begin{aligned}
			E_{1}=&\{\tilde{\theta}_{n}^{(j)} \in [z_{K}^{(j)}(\frac{\alpha}{2}-\upsilon), \;\;z_{K}^{(j)}(1-\frac{\alpha}{2}+\upsilon)]\},\\
			E_{2}=&\{\tilde{\theta}_{n}^{(j)}\in [z^{(j)}(\frac{\alpha}{2}),\;\;z^{(j)}(1-\frac{\alpha}{2})] \},\\
			E_{3}=&\{F^{(j)}_{K}(z_{K}^{(j)}(\frac{\alpha}{2}-\upsilon)-F^{(j)}(z_{K}^{(j)}(\frac{\alpha}{2}-\upsilon)))> -\upsilon \}\cap\\ &\{F^{(j)}_{K}(z_{K}^{(j)}(1-\frac{\alpha}{2}+\upsilon)-F^{(j)}(z_{K}^{(j)}(1-\frac{\alpha}{2}+\upsilon)))< \upsilon \}.
		\end{aligned}
	\end{equation}
	For every $\omega \in E_{3}$, we have
	\begin{equation}
		\begin{aligned}
			F^{(j)}(z_{K}^{(j)}(\frac{\alpha}{2}-\upsilon)) &< F_{K}^{(j)}(z_{K}^{(j)}(\frac{\alpha}{2}-\upsilon))+\upsilon=\frac{\alpha}{2},\\
			F^{(j)}(z_{K}^{(j)}(1-\frac{\alpha}{2}+\upsilon)) &>  F^{(j)}_{K}(z_{K}^{(j)}(1-\frac{\alpha}{2}+\upsilon))-\upsilon\\
			&=1-\frac{\alpha}{2},
		\end{aligned}
	\end{equation}
	which means
	\begin{equation}
		\begin{aligned}
			z^{(j)}(\frac{\alpha}{2})\geq z_{K}^{(j)}(\frac{\alpha}{2}-\upsilon),\;
			z^{(j)}(1-\frac{\alpha}{2})\leq z_{K}^{(j)}(1-\frac{\alpha}{2}+\upsilon).
		\end{aligned}
	\end{equation}
	Hence, we have 
	\begin{equation}\label{sub}
		(E_{2}\cap E_{3}) \subset E_{1}. 
	\end{equation}
	From the definition of $E_{2}$, we have 
	\begin{equation}\label{e2}
		P(E_{2}^{c})<\alpha.
	\end{equation}
	We now prove that 
	\begin{equation}\label{e3}
		P(E_{3}^{c})<t.
	\end{equation}
	Since $\mathbb{E}[F_{K}^{(j)}(x)]=E[\frac{1}{K}\sum_{i=1}^{K}I_{\{H_{j}(X_{i})\leq x\}}]=F^{(j)}(x)$ for any $x \in R$, by Hoffeding’s inequality, we have 
	\begin{equation}
		\begin{aligned}
			P\{F_{K}^{(j)}(z_{K}^{(j)}(\frac{\alpha}{2}-\upsilon))-F^{(j)}(z_{K}^{(j)}(\frac{\alpha}{2}-\upsilon))\leq -\upsilon \}\leq \frac{t}{2}\\
			P\{F_{K}^{(j)}(z_{K}^{(j)}(1-\frac{\alpha}{2}+\upsilon)-F^{(j)}(z_{K}^{(j)}(1-\frac{\alpha}{2}+\upsilon))\geq \upsilon \}\leq \frac{t}{2}
		\end{aligned}
	\end{equation}
	Thus $(\ref{e3})$ holds true. By $(\ref{sub})$, $(\ref{e2})$ and $(\ref{e3})$, we finally have $P(E_{1})\geq 1-\alpha-t,$ which proves  Proposition $\ref{po1}$.
	\hspace*{\fill}~\QED\par\endtrivlist\unskip

	\section{Numerical simulation}\label{ss5}
In this section, we show that the parameter estimate in the second step of our TSQN algorithm outperforms that in the first step by an empirical analysis of the sentencing problem.  Notably,  the pronounced penalties are constrained within the statutory range of penalty according to the related basic criminal facts, they can be regarded as saturated output observations. The data source is taken from the China Judgements Online, where the data set was constructed according to judgment documents of the crime of intentional injury from 2011 to 2021. For more details about the modeling and the data explanation, see \cite{judical}.

Firstly, we compare the adaptive  prediction performance of the preliminary estimates $\bar{\theta}_{k}$ in the  first step with that in the second step  of the TSQN algorithm, based on the dataset of  sentencing for serious injury cases of intentional injury, where the prediction accuracy is defined by 
 $Predection-acc(t)=1-\frac{1}{T}\sum_{t=1}^{T}\frac{|y_{t}-\hat{y}_{t}|}{|y_{t}|}. $ 
From  Fig \ref{fig3}, one can see that the prediction accuracy in the second step outperforms significantly that in the first step.

Secondly, we compare the confidence intervals of the estimates obtained respectively by the first step and second step of the TSQN algorithm with finite data length. The confidence intervals are constructed through 10000 Monte Carlo simulations in accordance with Proposition 1. Our simulations reveal that the confidence interval of the estimates obtained in the second step is significantly smaller than that obtained in the first step. Fig \ref{fig4} shows the results of the first five components of the parameter vector, with $\theta_{1}$ corresponding to the feature ``Admitting guilt \& accepting punishments", $\theta_{2}$ corresponding to the feature ``Accessary criminal", $\theta_{3}$ corresponding to the feature ``Armed", $\theta_{4}$ corresponding to the feature ``Criminal record", $\theta_{5}$ corresponding to the feature ``Aged 75 and over". Similar results can be obtained for other components of the parameter vector, for details, see \cite{judical}.
\begin{figure}[htbp]
	\centerline{\includegraphics[width=9.5cm]{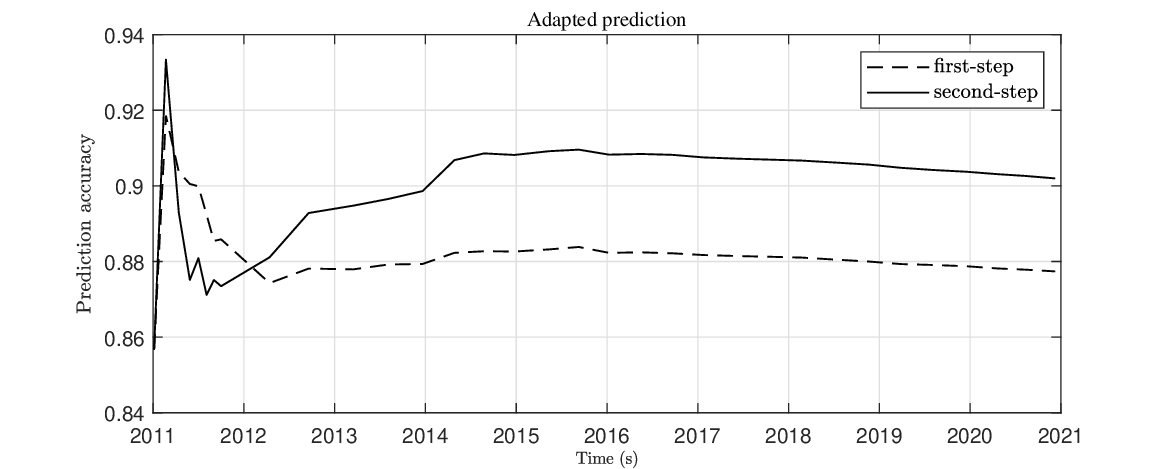}}
	\caption{Comparison of adaptive prediction accuracy.}
	\label{fig3}
\end{figure}
\begin{figure}[htbp]
	\centerline{\includegraphics[width=8.5cm]{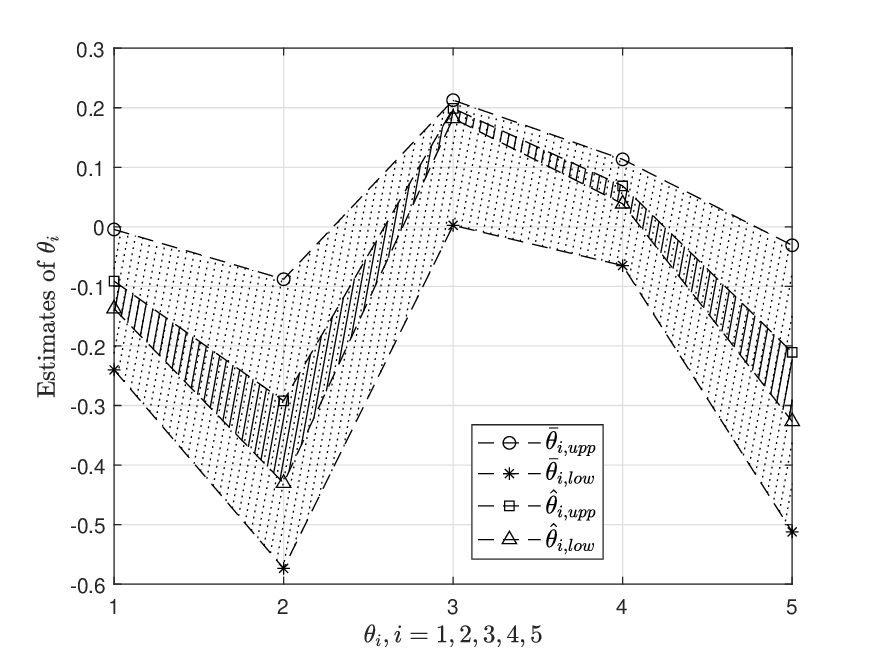}}
	\caption{Comparison of parameter estimates. $\bar{\theta}_{i,upp}$ and $\bar{\theta}_{i,low}$:  the upper and lower bounds of the 90\% confidence interval of the estimates produced by the first step; $\hat{\theta}_{i,upp}$ and $\hat{\theta}_{i,low}$: the upper and lower bounds of the 90\% confidence interval of the estimates produced by the second step; Shaded area: estimates within 90\% confidence interval.}
	\label{fig4}
\end{figure}

	\section{Conclusion}\label{ss6}
	Motivated by various application backgrounds, we have in this paper studied the problem of adaptive identification and prediction problems of stochastic dynamical systems with saturated observations. To improve the performance of the estimation algorithm designed naturally by using an adaptive single-step quasi-Newton method,  we have proposed a new adaptive two-step quasi-Newton algorithm to estimate the unknown parameters. It is shown that the strong consistency and the asymptotic normality of the estimate can be established under general non-PE conditions as the data length increases to infinity. When the data length is given and finite, it is also shown that the estimation performance can also be guaranteed with high probability by using either the Lyapunov function-based method or the Monte Carlo-based method, which appears to be more suitable for application problems where only data set with finite length is available. The numerical example also demonstrates that the performance of the proposed TSQN algorithm is better than the single-step quasi-Newton algorithm even under non-PE conditions of the data. The proposed new TSQN algorithm has also been used successfully in sentencing computation problems with real finite data set in \cite{judical} as outlined in Section V. For future investigations, there are still several interesting problems that need to be solved in theory, for example,  how to establish global convergence or estimation error bounds for adaptive estimation algorithms of more complicated stochastic regression models including multi-layer neural networks, and how to solve adaptive control problems with saturated observations for stochastic dynamical control systems, etc.

	\appendices
	
	%%%%%%%%%%%
	
	\section{}\label{i}
		In this appendix, we give three examples for the calculation of the functions $G_{k}(\cdot)$ and $\sigma_{k}(\cdot)$.
	\begin{example}
	Let us consider the case where  $l_{k}=-\infty$ and $u_{k}=\infty$ for any $k\geq 1$, then the model degenerate to the classical linear regression model, and in this case, we have $G(x)=\mathbb{E}_{k}[S(x+e_{k+1})]\equiv x,$ and  $\sigma_{k}(x)=\mathbb{E}_{k}\left[|S(x+e_{k+1})-G(x)|^{2}\right]=\mathbb{E}[e_{k+1}^{2}\mid \mathcal{F}_{k}].$
	\end{example}
    \begin{example}
	Let us consider the case where $L_{k}=l_{k}=u_{k}=0, U_{k}=1$, then the saturated function will turn to be a binary-valued function, which is widely used in classification problem. Let the noise $e_{k}$ is $\mathcal{F}_{k}-$measurable with the conditional probability distribution function $F_{k}(\cdot)$, then we have $G_{k}(x)=1-F_{k}(-x)$ and $\sigma_{k}(x)=F_{k}(-x)[1-F_{k}(-x)].$
    \end{example}
    \begin{example}
	Let us consider the case where  $L_{k}=l_{k}<u_{k}=U_{k}$ for any $k\geq 1$,  the noise sequence $\left\{e_{k}\right\}$ is independent and normally distributed with $e_{k} \sim N(0,\sigma^{2})$. Let the conditional probability distribution function and the conditional probability density function of $e_{k}$ be $F(\cdot)$ and $f(\cdot)$, respectively. Then the function $G_{k}(\cdot)$ and $\sigma_{k}(\cdot)$ can be calculated as follows:
		\begin{equation}
		\begin{aligned}
			G_{k}(x)=&u_{k}+(l_{k}-x)F(l_{k}-x)-(u_{k}-x)F(u_{k}-x)\\
			&+\sigma^{2}[f(l_{k}-x)-f(u_{k}-x)],\\
			\sigma_{k}(x)=&(u_{k}-G_{k}(x))^{2}+\sigma^{2}F(u_{k}-x)\\
			&+[(G_{k}(x)-x)^2-(u_{k}-G_{k}(x))^2]
			F(u_{k}-x)\\
			&-[(l_{k}-G_{k}(x))^2-(G_{k}(x)-x)^2-\sigma^{2}]F(l_{k}-x)\\
			&+\sigma^{2}[(l_{k}+x-2G_{k}(x))f(l_{k}-x)\\
			&-(u_{k}+x-2G_{k}(x))f(u_{k}-x)].	
		\end{aligned}	
	\end{equation}		
	\end{example}
	
	\section{}\label{ii}
	\begin{lemma}\label{lem1} (\cite{ce2001}). The projection operator given by Definition $\ref{def2}$  satisfies
		\begin{equation}\label{17}
			\|\Pi_{Q}(x)-\Pi_{Q}(y)\|_{Q} \leq \|x-y\|_{Q}\quad \forall x, y\in \mathbb{R}^{m}
		\end{equation}
	\end{lemma}
	
	\begin{lemma}\label{lem2} (\cite{cg1991}). Let $\left\{w_{n}, \mathcal{F}_{n}\right\}$ be a martingale difference sequence and $\left\{f_{n}, \mathcal{F}_{n}\right\}$ an adapted sequence.
		If $\sup_{n} 	\mathbb{E}[|w_{n+1}|^{\alpha}\mid \mathcal{F}_{n}] < \infty, a.s.$, for some $\alpha \in (0, 2]$, then as $n\rightarrow \infty$:
		\begin{equation}\label{19}
			\sum_{i=0}^{n}f_{i}w_{i+1} = O(s_{n}(\alpha)\log^{\frac{1}{\alpha}+\eta}(s_{n}^{\alpha}(\alpha)+e))\;a.s., \forall \eta >0,
		\end{equation}
		where $s_{n}(\alpha)=\left(\sum_{i=0}^{n}|f_{i}|^{\alpha}\right)^{\frac{1}{\alpha}}$.
	\end{lemma}
	\begin{lemma}\label{lem3} (\cite{lw1982}). Let \;$X_{1}, X_{2},\cdots$ be a sequence of vectors in $	\mathbb{R}^{m} (m\geq 1)$ and let $A_{n} = A_{0}+\sum_{i=1}^{n}X_{i}X_{i}^{\top}$. Let $|A_{n}|$ denote the determinant of $A_{n}$. Assume that $A_{0}$ is nonsingular, then as $n\rightarrow \infty$
		\begin{equation}\label{20}
			\sum_{k=1}^{n}\frac{X_{k}^{\top}A_{k-1}^{-1}X_{k}}{1+X_{k}^{\top}A_{k-1}^{-1}X_{k}}\leq \log|A_{n}|+\log|A_{0}|.	
		\end{equation}
	\end{lemma}

	\begin{lemma}\label{lem9}(\cite{y2003})
        For each $n\geq1$, let $\{S_{n,j}=\sum_{i=1}^{j}X_{ni}, \mathcal{F}_{j}, 1\leq j\leq n<\infty\}$ be an $\mathcal{L}_{2}$ stochastic sequence on $(\Omega, \mathcal{F}, p)$ satisfying 
        \begin{equation}
        \sum_{k=1}^{n} \mathbb{E}\left(X_{n, k} \mid \mathcal{F}_{ k-1}\right) \underset{n \rightarrow \infty}{\stackrel{p}{\longrightarrow}} 0, 
                \end{equation}
                        \begin{equation}
        \sum_{k=1}^{n} [\mathbb{E}\left(X_{n, k}^{2} \mid \mathcal{F}_{ k-1}\right)-(\mathbb{E}\left(X_{n, k} \mid \mathcal{F}_{ k-1}\right))^{2}] \underset{n \rightarrow \infty}{\stackrel{p}{\longrightarrow}} \mathcal{\eta}^{2}, 
        \end{equation}
        \begin{equation}
          \sum_{k=1}^{n} \mathbb{E}\left[X_{n, k}^{2} I\left(\left|X_{n, k}\right|>\varepsilon\right) \mid \mathcal{F}_{n, k-1}\right] \underset{n \rightarrow \infty}{\stackrel{p}{\longrightarrow}} 0,\;\;\epsilon>0, 
        \end{equation}
        for some non-negative constant $\mathcal{\eta}^{2}$, then we have $S_{n, n} \underset{n \rightarrow \infty}{\stackrel{d}{\longrightarrow}}N(0, \mathcal{\eta}^{2}). $
	\end{lemma}
	\begin{lemma}\label{lem8}(\cite{y2003})
		If $\{S_{n}=\sum_{k=1}^{n}X_{k},\mathcal{F}_{n},n\geq 1\}$ is an $\mathcal{L}_{2}$ martingale with $\mathbb{E}[S_{1}]=0$ and $\mathcal{F}_{0}=( \emptyset, \Omega)$, then for any positive constants $x, y$
		\begin{equation}
			P\{S_{n}\geq x\sum_{k=1}^{n}\mathbb{E}[X_{k}^{2}\mid \mathcal{F}_{k-1}]+y, some\;\; n\geq 1\}\leq \frac{1}{1+xy}.
		\end{equation}
		\end{lemma}
	\subsection{Proof of Remark $\ref{re7}$, $(\ref{ssss})$ and $(\ref{muu})$}\label{A}
	\noindent\hspace{2em}{\itshape Proof of Remark $\ref{re7}$: }	
		From $(\ref{kkk})$ in the proof of Theorem $\ref{thm3}$, we have
		\begin{equation}\label{ite}
			\begin{aligned}
				\tilde{\theta}_{k+1}=&P_{k+1}P_{0}^{-1}\tilde{\theta}_{0}-P_{k+1}\sum_{i=0}^{k}\mu_{i}^{-1}\beta_{i}\phi_{i}w_{i+1}\\
				&-P_{k+1}^{\frac{1}{2}}\sum_{i=0}^{k}P_{k+1}^{\frac{1}{2}}\mu_{i}^{-1}\beta_{i}(\xi_{i}-\beta_{i})\phi_{i}\phi_{i}^{\top}\tilde{\theta}_{i}\\
               &+P_{k+1}\sum_{i=0}^{k}P_{i+1}^{-1}(s_{i}-\Pi_{P_{i+1}^{-1}}\{s_{i}\})I_{\{s_{i}\not\in D_{i}\}},\;\;a.s.
			\end{aligned}
		\end{equation}
		We now analyze the RHS of $(\ref{ite})$ term by term. Firstly, following the similar analysis as in $(\ref{k1})$ and $(\ref{k2})$, we have
		\begin{equation}\label{93}
			\|P_{k+1}P_{0}^{-1}\tilde{\theta}_{0}\|=O(\frac{1}{k}),\;a.s.
		\end{equation} 
		\begin{equation}\label{94}
		\|P_{k+1}\sum_{i=0}^{k}P_{i+1}^{-1}(s_{i}-\Pi_{P_{i+1}^{-1}}\{s_{i}\})I_{\{s_{i}\not\in D_{i}\}}\|=O(\frac{1}{k}),\;a.s.
	\end{equation} 
		Secondly, since $\|\mu_{i}^{-1}\beta_{i}\phi_{i}\|$ is bounded, from a refined martingale estimation theorem (see \cite{w1985}, \cite{g2020}), we have $\|\sum_{i=0}^{k}\mu_{i}^{-1}\beta_{i}\phi_{i}w_{i+1}\|=O(\sqrt{k\log\log k})$. Hence the second term of $(\ref{ite})$ satisfies
		\begin{equation}\label{99}
			\|P_{k+1}\sum_{i=0}^{k}\mu_{i}^{-1}\beta_{i}\phi_{i}w_{i+1}\|=O\left(\sqrt{\frac{\log\log k}{k}}\right),\;\;a.s.
		\end{equation}
		For the third term of $(\ref{ite})$, by  Cauchy-Schwarz inequality, we have
		\begin{equation}\label{popo}
				\begin{aligned}
			&\|\sum_{i=0}^{k}P_{k+1}^{\frac{1}{2}}\mu_{i}^{-1}\beta_{i}(\xi_{i}-\beta_{i})\phi_{i}\phi_{i}^{\top}\tilde{\theta}_{i}\|^{2}\\
			\leq&tr[\sum_{i=0}^{k}\mu_{i}^{-1}\beta_{i}^{2}P_{k+1}^{\frac{1}{2}}\phi_{i}\phi_{i}^{\top}P_{k+1}^{\frac{1}{2}}][\sum_{i=0}^{k}\mu_{i}^{-1}(\xi_{i}-\beta_{i})^{2}(\phi_{i}^{\top}\tilde{\theta}_{i})^{2}]\\
			=&O\left(\sum_{i=0}^{k}\mu_{i}^{-1}(\xi_{i}-\beta_{i})^{2}(\phi_{i}^{\top}\tilde{\theta}_{i})^{2}\right),\;\;a.s.
				\end{aligned}
		\end{equation} 
Notice that under condition $(\ref{re})$ and Lemma $\ref{lem6}$, we have
		\begin{equation}\label{lo}
			\|\tilde{\theta}_{k}\|^{2}=O\left(\frac{\log k}{k}\right),\;\;\|\tilde{\bar{\theta}}_{k} \|^{2}=O\left(\frac{\log k}{k}\right),\;\;a.s.
		\end{equation}
		Besides, from $(\ref{73})$, we have
		\begin{equation}\label{1333}
			(\xi_{i}-\beta_{i})^{2}(\phi_{i}^{\top}\tilde{\theta}_{i})^{2}\leq\rho^{2}\max((\phi_{i}^{\top}\tilde{\theta}_{i})^{4},(\phi_{i}^{\top}\tilde{\bar{\theta}}_{i})^{4}).
		\end{equation}
		Then from $(\ref{popo})$, $(\ref{lo})$ and $(\ref{1333})$,  we have
		\begin{equation}\label{128}
			\begin{aligned}
				&\|\sum_{i=0}^{k}P_{k+1}^{\frac{1}{2}}\mu_{i}^{-1}\beta_{i}(\xi_{i}-\beta_{i})\phi_{i}\phi_{i}^{\top}\tilde{\theta}_{i}\|^{2}\\
				=&O(\sum_{i=0}^{k}\rho^{2}\max((\phi_{i}^{\top}\tilde{\theta}_{i})^{4},(\phi_{i}^{\top}\tilde{\bar{\theta}}_{i})^{4}))
				=O(1),\;\;a.s.\\
			\end{aligned}
		\end{equation}
		Hence the third term of $(\ref{ite})$ satisfies
		\begin{equation}\label{97}
			\|P_{k+1}^{\frac{1}{2}}\sum_{i=0}^{k}P_{k+1}^{\frac{1}{2}}\mu_{i}^{-1}\beta_{i}(\xi_{i}-\beta_{i})\phi_{i}\phi_{i}^{\top}\tilde{\theta}_{i}\|=O\left(\sqrt{\frac{1}{k}}\right),\;\;a.s.
		\end{equation}
		Therefore, Remark $\ref{re7}$ is true by $(\ref{93})$, $(\ref{94})$,  $(\ref{99})$ and $(\ref{97})$.
		\hspace*{\fill}
~\QED\par\endtrivlist\unskip

\noindent\hspace{2em}{\itshape Proof of $(\ref{ssss})$:}	
	 Let us prove that for any $0\leq t\leq 2+\eta$,  
	 \begin{equation}\label{wwww}
	 	\sup\limits_{|x|\leq M_{k}, k\geq 0}\mathbb{E}_{k}[|S_{k}(x+e_{k+1})-G_{k}(x)|^{t}]<\infty, \;\;a.s.
	 \end{equation}
	 By Assumption $\ref{assum2}$, for any $|x|\leq M_{k}$, we have $|S_{k}(x+e_{k+1})-S_{k}(x)|\leq |e_{k+1}|+2c.$ Hence,
	\begin{equation}\label{ppppp}
		\begin{aligned}
			\mathbb{E}_{k}\left[|S_{k}(x+e_{k+1})-S_{k}(x)|^{t}\right]
			= O\left(\mathbb{E}_{k}\left[|e_{k+1}|^{t}\right]\right),\;\;\;a.s.
		\end{aligned}
	\end{equation}
		\begin{equation}\label{000}
		\begin{aligned}
			\mathbb{E}_{k}\left[|S_{k}(x)-G_{k}(x)|^{t}\right]&=\mathbb{E}_{k}\left[\left|\mathbb{E}_{k}[S_{k}(x)-S_{k}(x+e_{k+1})]\right|^{t}\right]\\
			&= O(\mathbb{E}_{k} \left[|e_{k+1}|^{t}\right]),\;\;\;a.s.
		\end{aligned}
	\end{equation}
	 Thus by Assumption $\ref{assum4}$, $(\ref{ppppp})$ and $(\ref{000})$, $(\ref{wwww})$ is true.
Let $t=2$ and notice that $\sup\limits_{|x|\leq M_{k}, k\geq 0}\{G_{k}(x)\}<\overline{g}_{k}M_{k}<\infty,\;a.s.,$ $(\ref{ssss})$ holds. 	
		\hspace*{\fill}~\QED\par\endtrivlist\unskip	

\noindent\hspace{2em}{\itshape Proof of $(\ref{muu})$:}
We first prove that 
\begin{equation}\label{contro}
	\inf\limits_{|x|\leq M_{k}, k\geq 0}|G_{k}(x)-L_{k}|>0,\; a.s.
\end{equation}
by contradiction. If $(\ref{contro})$ were not true, there would exist a set $\mathcal{B}$ with positive probability, such that for any $\omega \in \mathcal{B}$, we have 
	$\inf\limits_{|x|\leq M_{k}, k\geq 0}|G_{k}(x)-L_{k}|=0$ on $\mathcal{B}$.
Thus, for any $\omega \in \mathcal{B}$, there would exist a sequence $\{x_{k}, |x_{k}|\leq M_{k}\}$ such that $G_{k}(x_{k})-L_{k}\rightarrow 0$ as $k\rightarrow \infty$. Since by $(\ref{iinf})$ the function $G_{k}(x)$ is a non-decreasing function of $x$ for any $k$, and  $G_{k}(x)\geq L_{k}$ for any $x \in \mathbb{R}$, we thus have $G_{k}(x_{k}-\epsilon)-L_{k}\rightarrow 0$, where $\epsilon=M-\sup\limits_{k \geq0}\{M_{k}\}$. Therefore, $\liminf\limits_{k\rightarrow \infty}\frac{|G_{k}(x_{k})-G_{k}(x_{k}-\epsilon)|}{\epsilon}=0$, which means $\inf\limits_{|x|\leq M, k\geq 0}G'_{k}(x)=0$ on $\mathcal{B}$ and thus controdicts with $(\ref{iinf})$. Therefore, $(\ref{contro})$ is true. Similarly, we  will have 
$
	\inf\limits_{|x|\leq M_{k}, k\geq 0}|G_{k}(x)-U_{k}|>0,\; a.s.
$
Hence, there exists a variable $\delta>0,\;a.s.,$ such that
\begin{equation}
	|G_{k}(x_{k})-U_{k}|>\delta,\;\;	|G_{k}(x_{k})-L_{k}|>\delta,\;\;a.s.
\end{equation}
Besides, let 
\begin{equation}\label{eps}
	\epsilon_{k}=\min\{(\frac{\underline{e}_{k}^{2}}{3\overline{e}_{k}^{2}})^{\frac{2+\eta}{\eta}}, (\frac{\underline{e}_{k}^{2}}{3m_{k}\overline{e}_{k}})^{2}\},
\end{equation}
where $m_{k}=\sup\limits_{0\leq j\leq k}2(1+\overline{g}_{j})M_{j}$,  $\underline{e}_{k}^{2}=\inf\limits_{0\leq j\leq k} \mathbb{E}_{j}\{e_{j+1}^{2}\},$ and $\overline{e}_{k}=\sup\limits_{0\leq j\leq k}[\mathbb{E}_{j}\{|e_{j+1}|^{2+\eta}\}]^{\frac{1}{2+\eta}}$. Then from Assumption $\ref{assum4}$, we  have 
\begin{equation}\label{epsil}
\inf\limits_{k\geq 0}\{\epsilon_{k}\}>0,\;a.s.
\end{equation}
Moreover,  denote a sequence of minima $\{x_{k}, |x_{k}|\leq M_{k}\}$ such that $\sigma_{k}(x_{k})=\inf\limits_{|x|\leq M_{k}}\sigma_{k}(x),\;a.s.,$ we can easily know that $x_{k}$ is $\mathcal{F}_{k}-$measurable, we now show that
\begin{equation}\label{iii}
\sigma_{k}(x_{k})>\min\{\frac{1}{3}\underline{e}_{k}^{2}, \delta^{2}\epsilon_{k}\},\;a.s.,
\end{equation}
which will give the desirable result $(\ref{muu})$ by Assumption $\ref{assum4}$ and $(\ref{epsil})$. Indeed, let $\mathcal{A}_{k}=\{\omega: l_{k}\leq x_{k}+e_{k+1} \leq u_{k}\}$ and $\mathcal{B}_{k}=\{\omega: \mathbb{E}_{k}[I_{\mathcal{A}_{k}^{c}}]\geq \epsilon_{k}\}$.
By Cauchy-Schwarz inequality and $(\ref{eps})$, we have for any $k\geq 0$,
\begin{equation}\label{1433}
		\begin{aligned}
	&\mathbb{E}_{k}[e_{k+1}^{2}I_{\mathcal{A}_{k}^{c}}I_{\mathcal{B}_{k}^{c}}]\\ \leq&\left(\mathbb{E}_{k}[|e_{k+1}|^{2+\eta}]\right)^{\frac{2}{2+\eta}} \cdot \left(\mathbb{E}_{k}[I_{\mathcal{A}_{k}^{c}}]\right)^{\frac{\eta}{2+\eta}}\cdot I_{\mathcal{B}_{k}^{c}}\leq \frac{1}{3}\underline{e}_{k}^{2}I_{\mathcal{B}_{k}^{c}},\;a.s.\\
	&\mathbb{E}_{k}[e_{k+1}I_{\mathcal{A}_{k}^{c}}I_{\mathcal{B}_{k}^{c}}]\\ 
	\leq&\left(\mathbb{E}_{k}[|e_{k+1}|^{2}]\right)^{\frac{1}{2}} \cdot \left(\mathbb{E}_{k}[I_{\mathcal{A}_{k}^{c}}]\right)^{\frac{1}{2}}I_{\mathcal{B}_{k}^{c}}\leq \frac{1}{3m_{k}}\underline{e}_{k}^{2} I_{\mathcal{B}_{k}^{c}},\;a.s.
		\end{aligned}
\end{equation}
Therefore, by $(\ref{1433})$
we have for any $k$,
\begin{equation}\label{aaa}
		\begin{aligned}
	&\sigma_{k}(x_{k})=\mathbb{E}_{k}\left[|S_{k}(x_{k}+e_{k+1})-G_{k}(x_{k})|^{2}\right]\\
	\geq&\mathbb{E}_{k}\left[|S_{k}(x_{k}+e_{k+1})-G_{k}(x_{k})|^{2}I_{\mathcal{A}_{k}}I_{\mathcal{B}_{k}^{c}}\right]\\
	&+\mathbb{E}_{k}\left[|S_{k}(x_{k}+e_{k+1})-G_{k}(x_{k})|^{2}I_{\mathcal{A}_{k}^{c}}I_{\mathcal{B}_{k}}\right]\\
	\geq &\mathbb{E}_{k}\left[|x_{k}+e_{k+1}-G_{k}(x_{k})|^{2}I_{\mathcal{A}_{k}}I_{\mathcal{B}_{k}^{c}}\right]+\delta^{2}\mathbb{E}_{k}\left[I_{\mathcal{A}_{k}^{c}}\right]I_{\mathcal{B}_{k}}\\
	\geq&\mathbb{E}_{k}[e_{k+1}^{2}]I_{\mathcal{B}_{k}^{c}}-\mathbb{E}_{k}[e_{k+1}^{2}I_{\mathcal{A}_{k}^{c}}I_{\mathcal{B}_{k}^{c}}]+\delta^{2}\epsilon_{k}I_{\mathcal{B}_{k}}\\
	&+2|x_{k}-G_{k}(x_{k})|\mathbb{E}_{k}[e_{k+1}I_{\mathcal{A}_{k}^{c}}I_{\mathcal{B}_{k}^{c}}]\\
	\geq & \frac{1}{3}\underline{e}_{k}^{2}I_{\mathcal{B}_{k}^{c}}+\delta^{2}\epsilon_{k}I_{\mathcal{B}_{k}}\geq \min\{\frac{1}{3}\underline{e}_{k}^{2}, \delta^{2}\epsilon_{k}\},
	\;\;\;\;a.s.
		\end{aligned}
\end{equation}
Therefore, $(\ref{iii})$ is true and $(\ref{muu})$ is finally obtained.
	\hspace*{\fill}~\QED\par\endtrivlist\unskip


\begin{thebibliography}{00}
		
		\bibitem{mc1943}
		W. S. McCulloch and W. Pitts,
		 ``A logical calculus of the ideas immanent in nervous activity,'' 
		  \emph{The Bulletin of Mathematical Biophysics,} vol. 5, no. 10, pp. 115-133,  Dec. 1943. 
	
	   	\bibitem{gs1990}
	    S. I. Gallant,
	    ``Perceptron-based learning algorithms,''
	    \emph{IEEE Transactions on Neural Networks,} vol. 1, no. 2,  pp. 179-191, Jun. 1990.
	
	
              \bibitem{sj2004}
		J. Sun, Y. W. Kim and L. W,
		``Aftertreatment control and adaptation for automotive lean burn engines with HEGO sensors,"
		\emph{International Journal of Adaptive Control and Signal Processing},
		vol. 18, no. 2, pp. 145-166, Mar. 2004.



		\bibitem{hf2009}
		H. F. Grip~\emph{et al.},
		``Vehicle sideslip estimation,"
		\emph{IEEE Control Systems Magazine,
		vol. 29, no. 5, pp. 36-52, Oct. 2009.}
		
		\bibitem{tj1958}
		J. Tobin,
		``Estimation of relationships for limited dependent variables,''
		\emph{ Econometrica: Journal of the Econometric Society}, vol. 26, no. 1, pp. 24-36, Jan. 1958.

		
		\bibitem{mj2020}
		M. S. Jeon and J. H. Lee, 
		``Estimation of willingness-to-pay for premium economy class by type of service,"
		\emph{ Journal of Air Transport Management}, vol. 84, May. 2020.
		
		\bibitem{cm2013}
		A. Bykhovskaya, 
		``Time series approach to the evolution of networks: prediction and estimation,'' 
		\emph{Journal of Business $\&$ Economic Statistics}, vol. 41, pp. 170-183, 2023.
				
		\bibitem{judical}
		F. Wang,  L. Zhang and  L. Guo, 
		``Applications of nonlinear recursive identification theory in sentencing data analyses,''  
		\emph{SCIENTIA SINICA Informationis},
		vol. 52, no. 10, 2022, https://doi.org/10.1360/SSI-2022-0325.

		\bibitem{cg1991}
H. F. Chen and L. Guo, 
\emph{Identification and Stochastic Adaptive Control,}
Birkh{\"a}suser, Boston, 1991.

		\bibitem{al2021}
		A. L. Bruce, A. Goel and D. S. Bernstein, 
		``Necessary and sufficient regressor conditions for the global asymptotic stability of recursive least squares,''
		\emph{Systems $\&$ Control Letters}, vol. 157, Nov. 2021.

             \bibitem{ss2023}
		Shadab S et al., 
		``Finite-time parameter estimation for an online monitoring of transformer: A system identification perspective," 
		\emph{International Journal of Electrical Power $\&$ Energy Systems},  vol. 145,  Feb. 2023.

		\bibitem{Lj1977}
	     L. Ljung, 
		``Analysis of recursive stochastic algorithms,''
		\emph{ IEEE Transactions on Automatic Control},  vol. 22, no. 4, pp. 551-575, Aug. 1977.


		\bibitem{ll1976}
		L. Ljung, 
		``Consistency of the least-squares identification method,''
		\emph{  IEEE Transactions on Automatic Control}, vol. 21, no. 15, pp. 779-781, Oct. 1976.
		
		\bibitem{mj1978}
		J. B. Moore, 
		``On strong consistency of least squares identification algorithms,''
		\emph{  Automatica},  vol. 14, no. 5, pp. 505-509, Sep. 1978.
				
		\bibitem{lw1982}	
		T. L. Lai and  C. Z. Wei,
		``Least squares estimates in stochastic regression models with applications to identification and control of dynamic systems,''
		\emph{ The Annals of Statistic}, vol. 10, no. 1, pp. 154-166, Mar. 1982.
		
		\bibitem{hj1976}
		J. J. Heckman, 
	    ``The common structure of statistical models of truncation, sample selection and limited dependent variables and a simple estimator for such models,''
		\emph{ Annals of Economic and Social Measurement}, vol. 5, no. 4, pp. 475-492, Oct. 1976.

		
		\bibitem{1412020}
		L. Guo, 
		``Feedback and uncertainty: Some basic problems and results,''
	    \emph{ Annual Reviews in Control},  vol. 49, pp. 27-36, May. 2020.
		
		\bibitem{PJ1984}
		 J. L. Powell,
		``Least absolute deviations estimation for the censored regression model,''
		\emph{ Journal of Econometrics},  vol. 25, no. 3, pp. 303-325, Jul. 1984.
		
		\bibitem{ly1992}
		T. L. Lai and Z. Ying, 
		``Asymptotically efficient estimation in censored and truncated regression models,''
		\emph{ Statistica Sinica},  vol. 2, no. 1, pp. 17-46, Jan. 1992.
		
				
		
		\bibitem{ZG2003}
		L. Y. Wang, J. F. Zhang and  G. G. Yin, 
		``System identification using binary sensors,''
		\emph{ IEEE Transactions on Automatic Control},  vol. 48, no. 11, pp. 1892-1907, Nov. 2003.
		
		\bibitem{YZ2006}
		 L. Y. Wang, G. G. Yin and  J. F. Zhang,
		``Joint identification of plant rational models and noise distribution functions using binary-valued observation,''
		\emph{ Automatica},  vol. 42, no. 4, pp.  535-547, Apr. 2006.
		
		\bibitem{YG2007}
		 L. Y. Wang and G. G. Yin, 
		``Asymptotically efficient parameter estimation using quantized output observations,''
		\emph{ Automatica},  vol. 43, no. 7, pp. 1178-1191, Jul. 2007.
		
		\bibitem{GZ2013}
		J. Guo  and Y. Zhao, 
		``Recursive projection algorithm on FIR system identification with binary-valued observations,''
		\emph{Automatica}, vol. 49, no. 11, pp. 3396-3401, Nov. 2013.
		
		\bibitem{yy2015}
		Y. Guo~\emph{et al.},
		``Identification of nonlinear systems with non-persistent excitation using an iterative forward orthogonal least squares regression algorithm,"
		 \emph{International Journal of Modelling, Identification and Control}, vol. 23, no. 1, pp. 1-7, Apr. 2015.
				
		\bibitem{sm2018}
		S. Mei, Y. Bai and A. Montanari,
		``The Landscape of empirical risk for nonconvex losses,"
		\emph{The Annals of Statistics}, vol. 46, no. 6A, pp. 2747-2774, Dec. 2018.

	
				
		\bibitem{ZZ2022} 
		 L. Zhang,  Y. Zhao and  L. Guo,
		``Identification and adaptation with binary-valued observations under non-persistent excitation,''
		\emph{ Automatica}, vol. 138: 110158,  Apr. 2022. 
		

		 
		
\bibitem{kv2019}
	     V. Kontonis,  C. Tzamos and M. Zampetakis, 
		``Efficient truncated statistics with unknown truncation,'' 
		in \emph{IEEE 60th Annual Symposium on Foundations of Computer Science (FOCS)}, Baltimore, MD, USA,  Nov. 2019,  pp. 1578-1595.
		
		\bibitem{dc2021}
		C. Daskalakis~\emph{et al.}, 
		``Efficient truncated linear regression with unknown noise variance,''
		\emph{Advances in Neural Information Processing Systems}, vol. 34, pp. 1952-1963, 2021. 
		
		
				
		
		\bibitem{po2021} 
		O. Plevrakis, 
		``Learning from censored and dependent data: The case of linear dynamics,'' 
		in \emph{Proc. 34th Conference on Learning Theory,}  2021,  pp. 3771-3787.
		
	
		\bibitem{gl1993}
		L. Guo, L. Ljung  and P. Priouret,
		``Performance analysis of the forgetting factor RLS algorithm,'' 
		\emph{International Journal of Adaptive Control and Signal Processing}, vol. 7, no. 6, pp. 525-537, Nov. 1993.
		
		
		
		\bibitem{g1995}
		L. Guo,
		``Convergence and logarithm laws of self-tuning regulators,''  
		\emph{Automatica}, vol. 31, no. 3, pp. 435-450, Mar. 1995.
		
		\bibitem{g2020}
		L. Guo,  
		\emph{Time-Varying Stochastic Systems: Stability and Adaptive Theory, 2}nd ed. Beijing: Science Press, 2020.
		
			\bibitem{w1985}	
		C. Z. Wei,
		``Asymptotic properties of least-squares estimates in stochastic regression models,''  
		\emph{The Annals of Statistics}, vol. 13, no. 4, pp. 1498-1508, Dec. 1985.
		
		\bibitem{lai1986}	
		L. T. Lai,
		``Asymptotically efficient adaptive control in stochastic regression models,''  
		\emph{Advances in Applied Mathematics}, vol. 7, no. 1, pp. 23-45, Mar. 1986.
		
		
		\bibitem{sm2003}
		J. Shao, 
		\emph{Mathematical Statistics, 2}nd ed. 
		New York: Springer, 2003. 
		
	
		\bibitem{dp2014}
		 D. P. Kroese and J. C. C. Chan,
		\emph{Statistical Modeling and Computation,}  
		New York: Springer, 2014.
		
		
		\bibitem{ce2001}
		W. Cheney, 
		\emph{Analysis for Applied Mathematics,} 
		New York: Springer, 2001.
		 
		\bibitem{y2003}
		Y. S. Chow  and  H. Teicher,
		\emph{Probability Theory: Independence, Interchangeability, Martingales, 3}th ed. 
		New York: Springer, 1997.	
		
				
	\end{thebibliography}
	\end{document}